\documentclass[conference]{IEEEtran}%
\usepackage{amsfonts}
\usepackage{amsmath}
\usepackage{amssymb}
\usepackage{graphicx}
\usepackage[colorlinks=true,linkcolor=blue,citecolor=red,plainpages=false,pdfpagelabels]%
{hyperref}%
\providecommand{\U}[1]{\protect\rule{.1in}{.1in}}
\newtheorem{theorem}{Theorem}

\newtheorem{proposition}{Proposition}

\newenvironment{proof}[1][Proof]{\noindent\textbf{#1.} }{\ \rule{0.5em}{0.5em}}
\allowdisplaybreaks
\begin{document}

\title{Coherent Quantum Channel Discrimination}

%

\author{\IEEEauthorblockN{Mark M. Wilde}
\IEEEauthorblockA
{Hearne Institute for Theoretical Physics, Department of Physics and Astronomy\\Center for Computation and Technology, Louisiana State University\\Baton
Rouge, Louisiana 70803, USA,
Email: mwilde@lsu.edu}
}%
%

\maketitle
%

\begin{abstract}%

This paper introduces coherent quantum channel discrimination as a coherent
version of conventional quantum channel discrimination. Coherent channel
discrimination is phrased here as a quantum interactive proof system between a
verifier and a prover, wherein the goal of the prover is to distinguish two
channels called in superposition in order to distill a Bell state at the end.
The key measure considered here is the success probability of distilling a
Bell state, and I prove that this success probability does not increase under
the action of a quantum superchannel, thus establishing this measure as a
fundamental measure of channel distinguishability. Also, I establish some
bounds on this success probability in terms of the success probability of
conventional channel discrimination. Finally, I provide an explicit semi-definite
program that can compute the success probability.%

\end{abstract}%

\section{Introduction}

\label{sec:intro}

Quantum channel discrimination is a fundamental information-processing task in
quantum information theory
\cite{Kitaev1997,AKN98,CPR00,Acin01,RW05,GLN04,PhysRevA.71.062340,PhysRevA.72.014305}%
. There are at least two ways of thinking about it: one in terms of
quantifying error between an ideal channel and an experimental approximation
of it \cite{Kitaev1997,AKN98} and another in terms of symmetric hypothesis
testing \cite{RW05,GLN04}. In both scenarios, the diamond distance between
channels \cite{Kitaev1997} arises as the fundamental metric quantifying the
distinguishability of two quantum channels. These 
interpretations of diamond distance are the main reason that it is employed as
the primary theoretical quantifier of channel distance in applications such as
fault-tolerant quantum computation \cite{QECbook}, quantum complexity theory
\cite{VW16}, and quantum Shannon theory~\cite{KWerner04}.

To expand upon the first way of thinking about channel discrimination from
\cite{Kitaev1997,AKN98}, suppose that the ideal channel to be implemented is
$\mathcal{N}_{A\rightarrow B}^{0}$ (a completely positive, trace-preserving
map taking operators for a system$~A$ to operators for a system$~B$). Suppose
further that the experimental approximation is $\mathcal{N}_{A\rightarrow
B}^{1}$. To interface with these channels and obtain classical data, the most
general way for doing so is to prepare a state $\rho_{RA}$
of a reference system $R$ and the channel input system $A$, feed system~$A$
into the unknown channel, and then perform a quantum measurement
$\{\Lambda_{RB}^{x}\}_{x\in\mathcal{X}}$ on the channel output system $B$ and
the reference system $R$. To be a legitimate quantum measurement, the set
$\{\Lambda_{RB}^{x}\}_{x\in\mathcal{X}}$ of operators should satisfy
$\sum_{x\in\mathcal{X}}\Lambda_{RB}^{x}=I_{RB}$ and $\Lambda_{RB}^{x}\geq0$
for all $x\in\mathcal{X}$. The result of this procedure (preparation, channel
evolution, and measurement) is a classical outcome $x\in\mathcal{X}$ that
occurs with probability $\operatorname{Tr}[\Lambda_{RB}^{x}\mathcal{N}%
_{A\rightarrow B}^{0}(\rho_{RA})]$ if the channel $\mathcal{N}_{A\rightarrow
B}^{0}$ is applied, while the outcome $x\in\mathcal{X}$ occurs with
probability $\operatorname{Tr}[\Lambda_{RB}^{x}\mathcal{N}_{A\rightarrow
B}^{1}(\rho_{RA})]$ if the channel $\mathcal{N}_{A\rightarrow B}^{1}$ is
applied. The error or difference between these probabilities is naturally
quantified by the absolute deviation%
\begin{equation}
\left\vert \operatorname{Tr}[\Lambda_{RB}^{x}\mathcal{N}_{A\rightarrow B}%
^{0}(\rho_{RA})]-\operatorname{Tr}[\Lambda_{RB}^{x}\mathcal{N}_{A\rightarrow
B}^{1}(\rho_{RA})]\right\vert . \label{eq:ch-error-proc-depend}%
\end{equation}
We can then quantify the maximum possible error
between the channels $\mathcal{N}_{A\rightarrow B}^{0}$ and $\mathcal{N}%
_{A\rightarrow B}^{1}$ by optimizing \eqref{eq:ch-error-proc-depend} with
respect to all preparations and measurements:%
\begin{multline}
\sup_{\rho_{RA},\{\Lambda_{RB}^{x}\}_{x}}\left\vert \operatorname{Tr}%
[\Lambda_{RB}^{x}\mathcal{N}^{0}(\rho_{RA})]-\operatorname{Tr}[\Lambda
_{RB}^{x}\mathcal{N}^{1}(\rho_{RA})]\right\vert \label{eq:ch-error}\\
=\sup_{\substack{\rho_{RA},\\0\leq\Lambda_{RB}\leq I_{RB}}}\left\vert
\operatorname{Tr}[\Lambda_{RB}\mathcal{N}^{0}(\rho_{RA})]-\operatorname{Tr}%
[\Lambda_{RB}\mathcal{N}^{1}(\rho_{RA})]\right\vert ,
\end{multline}
where it is implicit that the channels $\mathcal{N}^{0}$ and $\mathcal{N}^{1}$
above have input system $A$ and output system $B$. Mathematically, this has
the effect of removing the dependence on the preparation and measurement such
that the error is a function solely of the two channels $\mathcal{N}%
_{A\rightarrow B}^{0}$ and $\mathcal{N}_{A\rightarrow B}^{1}$. It is a
fundamental and well known result in quantum information theory
\cite{Kitaev1997,AKN98}\ that the error in \eqref{eq:ch-error} is equal to the
normalized diamond distance:%
\begin{equation}
\text{Eq.}~\eqref{eq:ch-error}=\frac{1}{2}\left\Vert \mathcal{N}%
^{0}-\mathcal{N}^{1}\right\Vert _{\diamond},
\end{equation}
where the diamond distance $\left\Vert \mathcal{N}^{0}-\mathcal{N}%
^{1}\right\Vert _{\diamond}$ is defined as%
\begin{equation}
\left\Vert \mathcal{N}^{0}-\mathcal{N}^{1}\right\Vert _{\diamond}:=\sup
_{\psi_{RA}}\left\Vert \mathcal{N}_{A\rightarrow B}^{0}(\psi_{RA}%
)-\mathcal{N}_{A\rightarrow B}^{1}(\psi_{RA})\right\Vert _{1}.
\label{eq:def-diamond-distance}%
\end{equation}
In \eqref{eq:def-diamond-distance}, the optimization is with respect to all
pure bipartite states $\psi_{RA}$ with system $R$ isomorphic to the channel
input system $A$, and the trace norm of an operator $X$ is given by
$\left\Vert X\right\Vert _{1}=\operatorname{Tr}[\left\vert X\right\vert ]$,
where $\left\vert X\right\vert :=\sqrt{X^{\dag}X}$. This interpretation of
normalized diamond distance as error between channels is the main reason that
it is employed in applications like fault-tolerant quantum computation
\cite{QECbook}.

The other setting in which diamond distance arises is in the context of
symmetric hypothesis testing of quantum channels~\cite{RW05,GLN04}. We can
also refer to this as \textquotedblleft incoherent quantum channel
discrimination,\textquotedblright\ a name that shall become clear later. This
can be thought of as a guessing game between a prover and a verifier
\cite{RW05,R09}, and here we describe the game with fully quantum-mechanical
notation. Let us call it the \textquotedblleft channel guessing
game.\textquotedblright\ The game begins with the verifier flipping a fair
coin described by the state $\frac{1}{2}\left(  |0\rangle\langle0|_{R_{1}%
}+|1\rangle\langle1|_{R_{1}}\right)  $. Meanwhile the prover prepares a pure
state $\psi_{RA}$ and sends system $A$ to the verifier. The verifier then
performs the conditional channel%
\begin{equation}
\sum_{i\in\left\{  0,1\right\}  }|i\rangle\langle i|_{R_{1}}(\cdot
)|i\rangle\langle i|_{R_{1}}\otimes\mathcal{N}_{A\rightarrow B}^{i}(\cdot)
\label{eq:cond-chan}%
\end{equation}
on systems $R_{1}$ and $A$, so that the resulting global state is $\frac{1}%
{2}\sum_{i\in\left\{  0,1\right\}  }|i\rangle\langle i|_{R_{1}}\otimes
\mathcal{N}_{A\rightarrow B}^{i}(\psi_{RA})$. The verifier sends the channel
output system $B$ to the prover, whose task it is to guess which channel was
applied by the verifier. The prover can act on the systems in his possession,
which are $R$ and $B$. The prover performs a quantum-to-classical channel
$\sum_{j\in\left\{  0,1\right\}  }\operatorname{Tr}[\Lambda_{RB}^{j}%
(\cdot)]|j\rangle\langle j|_{F}$, where $\Lambda_{RB}^{j}\geq0$ for
$j\in\left\{  0,1\right\}  $ and $\sum_{j\in\left\{  0,1\right\}  }%
\Lambda_{RB}^{j}=I_{RB}$, and sends the system $F$ back to the verifier.
Finally, the verifier performs the measurement%
\begin{equation}
\left\{  |00\rangle\langle00|_{R_{1}F}+|11\rangle\langle11|_{R_{1}%
F},|01\rangle\langle01|_{R_{1}F}+|10\rangle\langle10|_{R_{1}F}\right\}
\label{eq:final-meas-incoh-ch-disc}%
\end{equation}
and declares \textquotedblleft success\textquotedblright\ if the first outcome
of the measurement occurs. If success occurs, we interpret this outcome as
meaning that the prover is able to distinguish the channels. Running through
the calculation, the probability that the prover wins (verifier declares
\textquotedblleft success\textquotedblright) is equal to%
\begin{equation}
\frac{1}{2}\sum_{i\in\left\{  0,1\right\}  }\operatorname{Tr}[\Lambda_{RB}%
^{i}\mathcal{N}_{A\rightarrow B}^{i}(\psi_{RA})].
\end{equation}
Figure~\ref{fig:q-ch-disc}\ depicts the channel guessing game (in order to
understand it fully, it is necessary to read the next
section).\begin{figure}[ptb]
\begin{center}
\includegraphics[
width=3.3399in
]{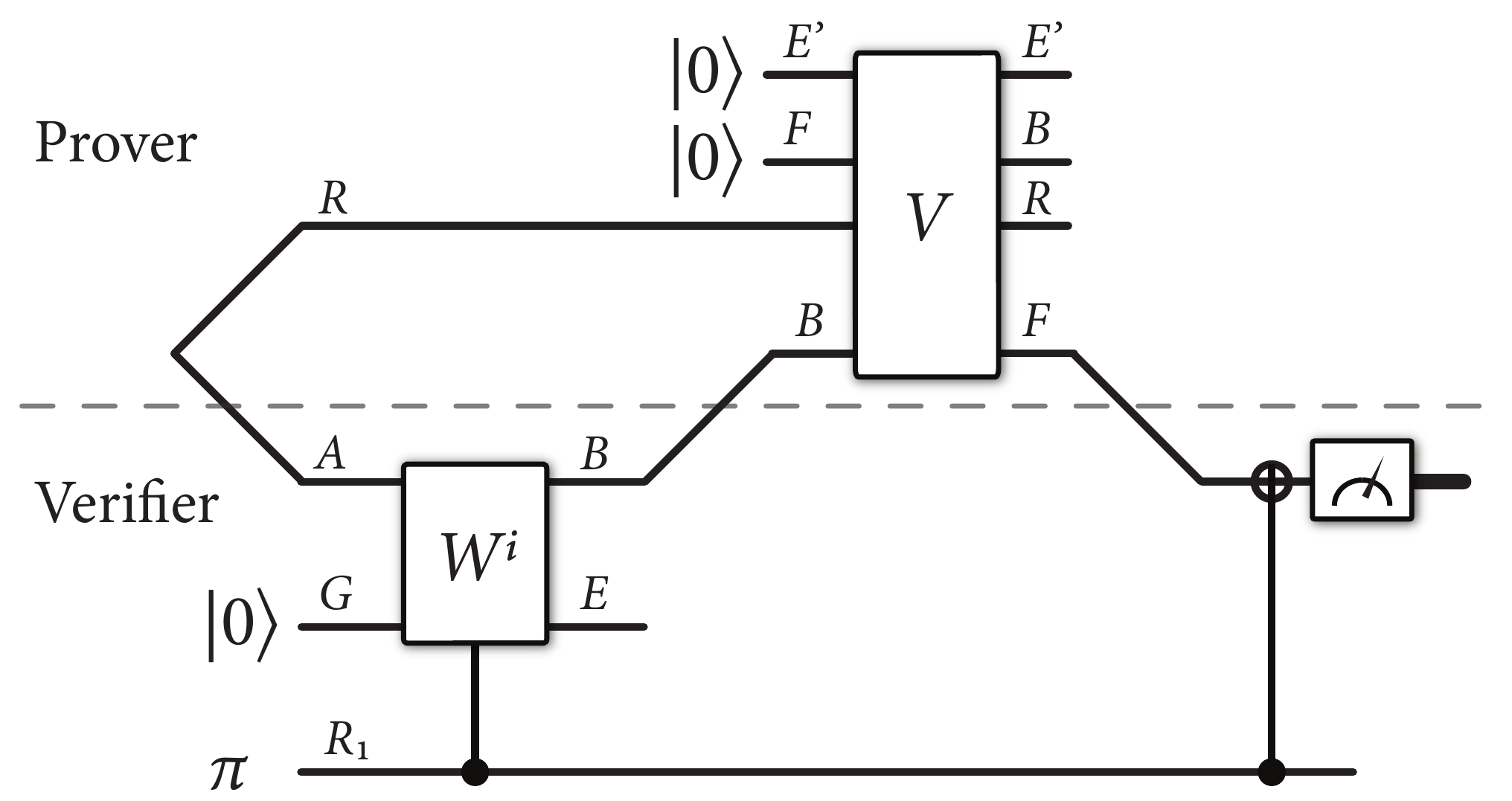}
\end{center}
\caption{In quantum channel discrimination, the prover prepares a pure state
$|\psi\rangle_{RA}$ and the verifier a mixed state $\pi:=\left(
|0\rangle\langle0|+|1\rangle\langle1|\right)  /2$. The verifier performs the
controlled unitary in \eqref{eq:controlled-unitary}\ that implements the
conditional channel in \eqref{eq:cond-chan}. The prover acts on the channel
output system $B$ and the reference system $R$ and sends back a single bit.
The final controlled-NOT\ and computational basis measurement implement the
measurement in \eqref{eq:final-meas-incoh-ch-disc}.}%
\label{fig:q-ch-disc}%
\end{figure}

The prover can optimize over all input states $\psi_{RA}$ and measurements
$\{\Lambda_{RB}^{j}\}_{j\in\left\{  0,1\right\}  }$, and a well known result
in quantum information \cite{RW05} is that the optimal success probability of
incoherent channel discrimination is given by%
\begin{align}
p_{s}^{\text{inc}}(\mathcal{N}^{0},\mathcal{N}^{1})  &  :=\sup_{\substack{\psi
_{RA},\\\{\Lambda_{RB}^{j}\}_{j\in\left\{  0,1\right\}  }}}\frac{1}{2}%
\sum_{i\in\left\{  0,1\right\}  }\operatorname{Tr}[\Lambda_{RB}^{i}%
\mathcal{N}^{i}(\psi_{RA})]\nonumber\\
&  =\frac{1}{2}\left(  1+\frac{1}{2}\left\Vert \mathcal{N}^{0}-\mathcal{N}%
^{1}\right\Vert _{\diamond}\right)  , \label{eq:incoh-succ-prob-ch-disc}%
\end{align}
thus endowing the normalized diamond distance $\frac{1}{2}\left\Vert
\mathcal{N}^{0}-\mathcal{N}^{1}\right\Vert _{\diamond}$ with another
operational meaning as the relative bias away from random guessing in a
channel guessing game of the above form. That is, a random guessing strategy
leads to a success probability of $1/2$ and can be employed when the channels
are the same or indistinguishable. However, when the channels have some
distinguishability so that $\frac{1}{2}\left\Vert \mathcal{N}^{0}%
-\mathcal{N}^{1}\right\Vert _{\diamond}\in(0,1]$, then the success probability
changes as a linear function of the normalized diamond distance and reaches
its peak value when the channels are orthogonal to each other (perfectly
distinguishable). This guessing game is a basic channel discrimination task in
quantum information theory and has found application in the setting of quantum
illumination \cite{PhysRevA.71.062340,PhysRevA.72.014305,L08}.

A useful fact about the diamond distance is that it can be computed by means
of a semi-definite program \cite{Wat09}:%
\begin{equation}
\frac{1}{2}\left\Vert \mathcal{N}^{0}-\mathcal{N}^{1}\right\Vert _{\diamond
}=\inf_{\mu,Z_{RB}\geq0}\left\{
\begin{array}
[c]{c}%
\mu:Z_{RB}\geq\Gamma_{RB}^{\mathcal{N}^{0}}-\Gamma_{RB}^{\mathcal{N}^{1}},\\
\mu I_{R}\geq\operatorname{Tr}_{B}[Z_{RB}]
\end{array}
\right\}  ,
\end{equation}
where $\Gamma_{RB}^{\mathcal{N}}:=\mathcal{N}_{A\rightarrow B}(\Gamma_{RA})$
is the Choi operator of the channel $\mathcal{N}_{A\rightarrow B}$, with
$\Gamma_{RA}:=|\Gamma\rangle\langle\Gamma|_{RA}$ and $|\Gamma\rangle
_{RA}:=\sum_{i}|i\rangle_{R}|i\rangle_{A}$, for orthonormal bases
$\{|i\rangle_{R}\}_{i}$ and $\{|i\rangle_{A}\}_{i}$. Thus, calculating the
diamond distance is efficient in the dimensions of the input $A$ and output
$B$.

\section{Coherent Quantum\ Channel\ Discrimination}

The main aim of the present paper is to introduce and analyze a fully quantum
or coherent version of the channel guessing game presented above. Let us call
it \textit{coherent quantum channel discrimination}, in contrast to the
incoherent channel discrimination task presented above. The primary
modification that I make to it is to replace all classical steps of the
verifier with their coherent counterparts, much like what was done previously
in \cite{prl2004harrow} to produce coherent versions of basic protocols in
quantum information such as superdense coding and teleportation (see also
\cite{Har09} in this context). The resulting protocol is related to the fully
quantum reading protocol from \cite{DBW17}. A recent series of works have
considered coherent control of quantum channels
\cite{AWHMB18,GRB19,PDEBL19,DNSM19}, but coherent quantum channel
discrimination is different from the protocols considered in these prior works.

I now briefly summarize coherent channel discrimination. The main idea is to
replace the initial state of the verifier with $|+\rangle_{R_{1}}:=\left(
|0\rangle_{R_{1}}+|1\rangle_{R_{1}}\right)  /\sqrt{2}$, the conditional
channel of the verifier with a controlled unitary, and the final measurement
with a projection onto the Bell state $|\Phi\rangle_{R_{1}F}:=\left(
|00\rangle_{R_{1}F}+|11\rangle_{R_{1}F}\right)  /\sqrt{2}$ (here and
throughout the rest of the paper, we refer to both state vectors and density
operators as states, as is conventional in the quantum information
literature). Later, we shall see that it is sensible to include an uncomputing
step to uncompute the controlled channel at the end before performing the Bell projection.

The modifications of the guessing game presented here could potentially have
applications in quantum computation, where gates are often promoted to
controlled gates and used in superposition. In particular, some works have
recently investigated the question of compiling quantum circuits on quantum
computers \cite{Khatri2019quantumassisted,SKCC19}. The coherent games
presented here could be used as benchmarks to assess how well an approximate
implementation of a circuit could be used instead of the ideal one, even when
it is employed in superposition (i.e., in controlled form).\ We do not
investigate this particular application here but instead leave it for future work.

Before presenting details of the coherent version of the channel guessing
game, let us recall some fundamental facts about quantum channels (see, e.g.,
\cite{W17book}). First, every quantum channel $\mathcal{N}_{A\rightarrow B}$
has a Kraus representation as $\mathcal{N}_{A\rightarrow B}(\rho_{A})=\sum
_{j}N^{j}\rho N^{j\dag}$, where $\left\{  N^{j}\right\}  _{j}$ is a set of
Kraus operators satisfying $\sum_{j}N^{j\dag}N^{j}=I_{A}$. Another fundamental
fact is that every quantum channel $\mathcal{N}_{A\rightarrow B}$ has an
isometric extension. That is, to every quantum channel $\mathcal{N}%
_{A\rightarrow B}$, there exists an isometry $U_{A\rightarrow BE}$ (satisfying
$U^{\dag}U=I_{A}$) such that $\mathcal{N}_{A\rightarrow B}(\rho_{A}%
)=\operatorname{Tr}_{E}[U_{A\rightarrow BE}(\rho_{A})(U_{A\rightarrow
BE})^{\dag}]$ for all input states $\rho_{A}$. Equivalently, there exists an
environment system $G$ and a unitary $W_{AG\rightarrow BE}$ such that%
\begin{equation}
\mathcal{N}_{A\rightarrow B}(\rho_{A})=\operatorname{Tr}_{E}[W_{AG\rightarrow
BE}(\rho_{A}\otimes|0\rangle\langle0|_{G})(W_{AG\rightarrow BE})^{\dag}].
\label{eq:unitary-extend}%
\end{equation}
Thus, we can set $U_{A\rightarrow BE}=W_{AG\rightarrow BE}|0\rangle_{G}$. Any
two isometric extensions of the original channel are related by an isometry
acting on the environment system $E$.

The coherent version of the channel guessing game proceeds as follows. The
verifier prepares the state $|+\rangle_{R_{1}}$ and the prover prepares
$|\psi\rangle_{RA}$. The prover sends the system $A$ to the verifier. The
verifier then adjoins the state $|0\rangle_{G}$ and performs the controlled
unitary%
\begin{equation}
\sum_{i\in\left\{  0,1\right\}  }|i\rangle\langle i|_{R_{1}}\otimes
W_{AG\rightarrow BE}^{i}, \label{eq:controlled-unitary}%
\end{equation}
where $W_{AG\rightarrow BE}^{i}$ is a unitary that extends the channel
$\mathcal{N}_{A\rightarrow B}^{i}$ as in \eqref{eq:unitary-extend}. Let
$U_{A\rightarrow BE}^{i}:=W_{AG\rightarrow BE}^{i}|0\rangle_{G}$ be the
corresponding isometric extension. The resulting state is then%
\begin{equation}
\frac{1}{\sqrt{2}}\sum_{i\in\left\{  0,1\right\}  }|i\rangle_{R_{1}%
}W_{AG\rightarrow BE}^{i}|\psi\rangle_{RA}|0\rangle_{G}.
\label{eq:state-verifier-after-1st-unitary}%
\end{equation}
The verifier transmits system $B$ to the prover, who then adjoins an
environment system $E^{\prime}$ in the state $|0\rangle_{E^{\prime}}$, a qubit
system $F$ in the state $|0\rangle_{F}$, and performs a unitary
$V_{RBE^{\prime}F}$. The resulting state is then%
\begin{equation}
\frac{1}{\sqrt{2}}\sum_{i\in\left\{  0,1\right\}  }|i\rangle_{R_{1}%
}V_{RBE^{\prime}F}W_{AG\rightarrow BE}^{i}|\psi\rangle_{RA}|000\rangle
_{GE^{\prime}F}.
\end{equation}
The prover sends systems $B$ and $F$ back to the verifier, who uncomputes the
controlled unitary in \eqref{eq:controlled-unitary} by performing%
\begin{equation}
\sum_{i\in\left\{  0,1\right\}  }|i\rangle\langle i|_{R_{1}}\otimes
W_{AG\rightarrow BE}^{i\dag}. \label{eq:inverse-unitary}%
\end{equation}
The state at this point is then%
\begin{equation}
\frac{1}{\sqrt{2}}\sum_{i\in\left\{  0,1\right\}  }|i\rangle_{R_{1}}W^{i\dag
}VW^{i}|\psi\rangle_{RA}|000\rangle_{GE^{\prime}F},
\end{equation}
where we omit system labels for brevity. The verifier finally performs the
measurement%
\begin{equation}
\left\{  \Phi_{R_{1}F}\otimes|0\rangle\langle0|_{G},I_{R_{1}FG}-\Phi_{R_{1}%
F}\otimes|0\rangle\langle0|_{G}\right\}  \label{eq:verifier-final-proj-meas}%
\end{equation}
on systems $R_{1}FG$, where $\Phi_{R_{1}F}\equiv|\Phi\rangle\langle
\Phi|_{R_{1}F}$, and declares \textquotedblleft success\textquotedblright\ (or
\textquotedblleft prover wins!\textquotedblright)\ if the first outcome
occurs. The probability of success is equal to%
\begin{multline}
p_{s}^{\text{coh}}(\mathcal{N}^{0},\mathcal{N}^{1},|\psi\rangle_{RA}%
,V_{RBE^{\prime}F})\label{eq:unopt-succ-prob}\\
:=\frac{1}{2}\left\Vert \langle\Phi|_{R_{1}F}\langle0|_{G}\sum_{i\in\left\{
0,1\right\}  }|i\rangle_{R_{1}}W^{i\dag}VW^{i}|\psi\rangle_{RA}|000\rangle
_{GE^{\prime}F}\right\Vert _{2}^{2}\\
=\frac{1}{2}\left\Vert \langle\Phi|_{R_{1}F}\sum_{i\in\left\{  0,1\right\}
}|i\rangle_{R_{1}}U^{i\dag}VU^{i}|\psi\rangle_{RA}|00\rangle_{E^{\prime}%
F}\right\Vert _{2}^{2},
\end{multline}
where the second expression follows from the fact that $U_{A\rightarrow
BE}^{i}=W_{AG\rightarrow BE}^{i}|0\rangle_{G}$. Figure~\ref{fig:coh-q-ch-disc}%
\ depicts coherent quantum channel discrimination.\begin{figure}[ptb]
\begin{center}
\includegraphics[
width=3.339in
]{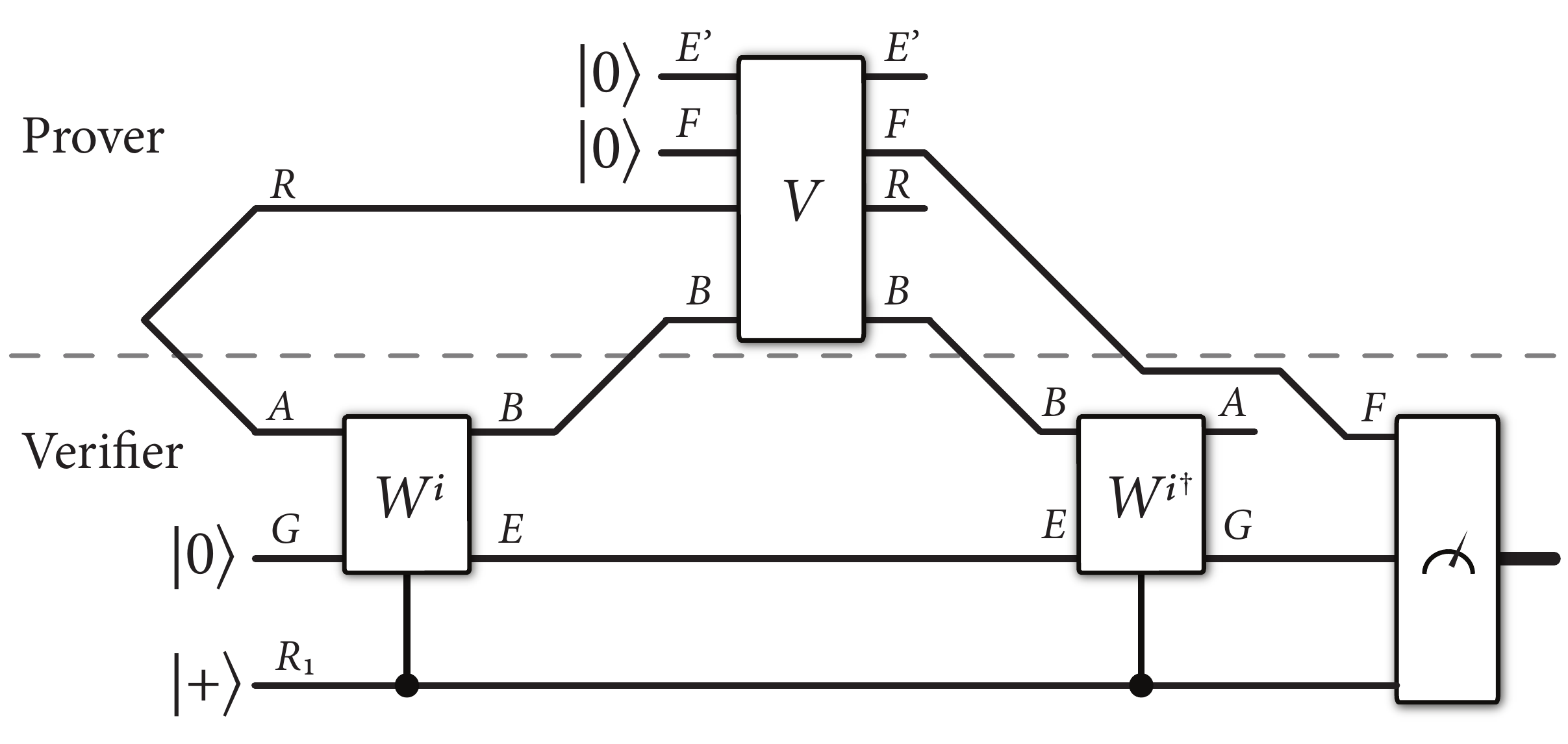}
\end{center}
\caption{In coherent quantum channel discrimination, the prover prepares a
pure state $|\psi\rangle_{RA}$ and the verifier the state $|+\rangle$. The
verifier performs the controlled unitary in \eqref{eq:controlled-unitary}. The
prover acts on the channel output system $B$ and reference system $R$ and
sends $B$ back along with a single qubit. The verifier uncomputes the
controlled unitary and finally implements the measurement in
\eqref{eq:verifier-final-proj-meas}.}%
\label{fig:coh-q-ch-disc}%
\end{figure}

We can already observe that the success probability in
\eqref{eq:unopt-succ-prob} is independent of the particular isometric
extension $U^{i}$ of the original channel $\mathcal{N}_{A\rightarrow B}^{i}$
for both $i=0$ and $i=1$. It is thus solely a function of the channels
$\mathcal{N}_{A\rightarrow B}^{0}$ and $\mathcal{N}_{A\rightarrow B}^{1}$, as
well as the particular strategy $\left\{  |\psi\rangle_{RA},V_{RBE^{\prime}%
F}\right\}  $ of the prover (as indicated by the notation in
\eqref{eq:unopt-succ-prob}). This follows because the unitary $V_{RBE^{\prime
}F}$ that the prover performs does not act on the environment system $E$.
Thus, letting $\widetilde{U}^{i}$ be some other isometric extension of
$\mathcal{N}_{A\rightarrow B}^{i}$, it follows that $\widetilde{U}^{i\dag
}V\widetilde{U}^{i}=U^{i\dag}VU^{i}$ by employing the previously stated fact
that there exists an isometry $T_{E^{\prime}}$ (satisfying $T^{\dag
}T=I_{E^{\prime}}$) such that $\widetilde{U}^{i}=T_{E^{\prime}}U^{i}$.

Just as in the guessing game presented in Section~\ref{sec:intro}, the prover
can optimize the success probability in \eqref{eq:unopt-succ-prob} with
respect to all possible strategies $\left\{  |\psi\rangle_{RA},V_{RBE^{\prime
}F}\right\}  $. Let us denote the resulting success probability as follows:%
\begin{equation}
p_{s}^{\text{coh}}(\mathcal{N}^{0},\mathcal{N}^{1}):=\sup_{|\psi\rangle
,V}p_{s}^{\text{coh}}(\mathcal{N}^{0},\mathcal{N}^{1},|\psi\rangle
_{RA},V_{RBE^{\prime}F}). \label{eq:succ-prob-coh}%
\end{equation}
The main goal of this paper is to understand this quantity in more detail and
relate it to the success probability in other forms of channel discrimination.

\section{Example}

As a very simple example to demonstrate the task of coherent channel
discrimination, suppose that the first channel $\mathcal{N}^{0}$ is the
identity channel and the second $\mathcal{N}^{1}$ is the deterministic
bit-flip channel, i.e., $\mathcal{N}^{1}(\cdot)=X(\cdot)X^{\dag}$, where $X=%
\begin{bmatrix}
0 & 1\\
1 & 0
\end{bmatrix}
$ is the Pauli flip operator. These channels are orthogonal to each other, and
a simple strategy for distinguishing them perfectly in incoherent channel
discrimination is to input the state $|0\rangle$ and perform a computational
basis measurement $\{|0\rangle\langle0|,|1\rangle\langle1|\}$. If the first
channel is applied, the output state is $|0\rangle$, while if the second
channel is applied, then the output state is $|1\rangle$, and these two states
are perfectly distinguishable.

For coherent channel discrimination, the same input state is optimal. To see
this, consider that the initial state of the verifier and prover's systems is
$|+\rangle_{R_{1}}|0\rangle_{A}$ (there is no reference system $R$ needed in
this case). The controlled unitary in \eqref{eq:controlled-unitary},
implemented by the verifier, is then a controlled-NOT gate $|0\rangle
\langle0|\otimes I+|1\rangle\langle1|\otimes X$, and there is no environment
system $E$ because the channels are unitary channels. The resulting state
after the controlled unitary is $|\Phi\rangle_{R_{1}B}$. The prover can then
perform a controlled-NOT\ gate from system~$B$ to system $F$, and the
resulting state is a GHZ\ state: $(|000\rangle_{R_{1}BF}+|111\rangle_{R_{1}%
BF})/\sqrt{2}$. The verifier then performs the inverse of the controlled-NOT
gate (itself a controlled-NOT), and the resulting state is $|\Phi
\rangle_{R_{1}F}|0\rangle_{B}$, so that the Bell projection at the end
succeeds with probability one; we thus arrive at the sensible conclusion that
these channels are perfectly distinguishable in coherent channel discrimination.

This key example illustrates the necessity and sensibility of the uncomputing
step in coherent channel discrimination. Without it, in this example, the
final Bell projection would succeed only with probability $1/2$, leading to
the unreasonable conclusion that these channels would not be perfectly
distinguishable in coherent channel discrimination. Uncomputing is commonly
employed in reversible and quantum computation as a \textquotedblleft
clean-up\textquotedblright\ step \cite{B73,B89,book2000mikeandike}, and it
serves the same purpose here.

\section{Results}

All proofs of the ensuing results appear in appendices.

\subsection{Alternate expression}

\begin{proposition}
\label{prop:coh-ch-disc-suc-prob-derive}For quantum channels $\mathcal{N}%
_{A\rightarrow B}^{0}$ and $\mathcal{N}_{A\rightarrow B}^{1}$, the success
probability in \eqref{eq:succ-prob-coh} is equal to%
\begin{multline}
p_{s}^{\text{coh}}(\mathcal{N}^{0},\mathcal{N}^{1}%
)=\label{eq:coh-ch-disc-suc-prob-derive}\\
\left[  \frac{1}{2}\sup_{\substack{\left\{  P^{i}\right\}  _{i\in\left\{
0,1\right\}  }:\\\sum_{i}P^{i\dag}P^{i}=I_{RB}}}\left\Vert \sum_{i\in\left\{
0,1\right\}  }(\mathcal{N}_{A\rightarrow B}^{i})^{\dag}(P_{RB\rightarrow
RBE^{\prime}}^{i})\right\Vert _{\infty}\right]  ^{2}.
\end{multline}
The operators $P_{RB\rightarrow RBE^{\prime}}^{i}$ act on the Hilbert space
for $RB$ and take them to the Hilbert space for $RBE^{\prime}$. The dimension
of $E^{\prime}$ need not be any larger than $2\left\vert A\right\vert
\left\vert B\right\vert ^{2}$.
\end{proposition}

In the above, the $\infty$-norm of an operator $X$ is defined as $\left\Vert
X\right\Vert _{\infty}=\sup_{|\varphi\rangle\neq0}\frac{\left\Vert
X|\varphi\rangle\right\Vert _{2}}{\left\Vert |\varphi\rangle\right\Vert _{2}}$
and the adjoint $\mathcal{N}^{\dag}$ of a quantum channel $\mathcal{N}$ is
defined to be the unique linear map satisfying $\operatorname{Tr}%
[[\mathcal{N}^{\dag}(Y)]^{\dag}X]=\operatorname{Tr}[Y^{\dag}\mathcal{N}(X)]$
for all operators $X$ and $Y$.

It is interesting to contrast the expression in
\eqref{eq:coh-ch-disc-suc-prob-derive} with the following expression for the
success probability of incoherent channel discrimination:%
\begin{equation}
p_{s}^{\text{inc}}(\mathcal{N}^{0},\mathcal{N}^{1})=\frac{1}{2}\sup_{\left\{
\Lambda_{RB}^{i}\right\}  _{i\in\left\{  0,1\right\}  }}\left\Vert \sum
_{i\in\left\{  0,1\right\}  }(\mathcal{N}_{A\rightarrow B}^{i})^{\dag}%
(\Lambda_{RB}^{i})\right\Vert _{\infty}, \label{eq:incoh-suc-prob-compare}%
\end{equation}
where $\sum_{i\in\left\{  0,1\right\}  }\Lambda_{RB}^{i}=I_{RB}$ and
$\Lambda_{RB}^{i}\geq0$. This expression comes about from that in
\eqref{eq:incoh-succ-prob-ch-disc}\ by employing the definition of the
$\infty$-norm and the adjoint of a quantum channel. Even by examining these
expressions, we can see how \eqref{eq:coh-ch-disc-suc-prob-derive} is a
coherent version of \eqref{eq:incoh-suc-prob-compare}. The expression in
\eqref{eq:coh-ch-disc-suc-prob-derive} is like the square of a probability
amplitude (the latter being the expression inside the $\infty$-norm), and it
involves operators for which the sum of their squares is equal to the identity
instead of their sums being equal to the identity.

\subsection{Bounds on success probability}

\begin{proposition}
\label{prop:bounds-suc-prob-1/2-1}The following bounds hold for the success
probability in \eqref{eq:succ-prob-coh}:%
\begin{equation}
1/2\leq p_{s}^{\text{coh}}(\mathcal{N}^{0},\mathcal{N}^{1})\leq1.
\label{eq:prob-suc-coh-disc-bnds}%
\end{equation}
The upper bound is saturated if and only if the channels are orthogonal (i.e.,
there exists a pure state $\psi_{RA}$ such that $\mathcal{N}_{A\rightarrow
B}^{0}(\psi_{RA})\mathcal{N}_{A\rightarrow B}^{1}(\psi_{RA})=0$). The lower
bound is saturated if the channels are identical (i.e., indistinguishable).
\end{proposition}

The upper bound is obvious since $p_{s}^{\text{coh}}$ is a probability, and
the necessary and sufficient condition for saturation follows by employing the
bounds%
\begin{equation}
p_{s}^{\text{coh}}(\mathcal{N}^{0},\mathcal{N}^{1})\leq p_{s}^{\text{inc}%
}(\mathcal{N}^{0},\mathcal{N}^{1})\leq\sqrt{p_{s}^{\text{coh}}(\mathcal{N}%
^{0},\mathcal{N}^{1})}, \label{eq:coh-inc-bounds}%
\end{equation}
discussed later. The lower bound follows by setting $P^{i}=\sqrt{1/2}%
I_{RB}\otimes|0\rangle_{E^{\prime}}$ for $i\in\left\{  0,1\right\}  $ in
\eqref{eq:coh-ch-disc-suc-prob-derive}, which corresponds to \textquotedblleft
not even trying to distinguish,\textquotedblright\ and the sufficient
saturation condition follows by direct evaluation.

\subsection{Non-increase under a superchannel}

\label{sec:nonincrease-superch}A key property of the success probability
$p_{s}^{\text{coh}}(\mathcal{N}^{0},\mathcal{N}^{1})$ in
\eqref{eq:succ-prob-coh} is that it does not increase under the action of a
quantum superchannel. This is a basic property expected of any channel
distinguishability measure, and it was recently shown that the diamond
distance (and thus the success probability in
\eqref{eq:incoh-succ-prob-ch-disc}) satisfies this property \cite{G18}.

To expand upon this statement, recall from \cite{CDP08} that a quantum
superchannel is a physical mapping of a quantum channel to a quantum channel,
and it should be this way even when acting on one share of an arbitrary
bipartite channel. In more detail, a superchannel $\Theta_{\left(
A\rightarrow B\right)  \rightarrow(C\rightarrow D)}$ is a linear map that
completely preserves the properties of complete positivity and trace
preservation. Then for an arbitrary input bipartite channel $\mathcal{M}%
_{RA\rightarrow RB}$, the output $\Theta_{\left(  A\rightarrow B\right)
\rightarrow(C\rightarrow D)}(\mathcal{M}_{RA\rightarrow RB})$ is a bipartite
channel from systems $RC$ to systems $RD$. The fundamental theorem of
superchannels is that any superchannel has a physical realization in terms of
a pre-processing channel $\mathcal{E}_{C\rightarrow AM}$ and a post-processing
channel $\mathcal{D}_{BM\rightarrow D}$ \cite{CDP08}:
\begin{multline}
\Theta_{\left(  A\rightarrow B\right)  \rightarrow(C\rightarrow D)}%
(\mathcal{M}_{RA\rightarrow RB})=\label{eq:FTOSCs}\\
\mathcal{D}_{BM\rightarrow D}\circ\mathcal{M}_{RA\rightarrow RB}%
\circ\mathcal{E}_{C\rightarrow AM}.
\end{multline}

With the fundamental theorem of superchannels in hand, we can arrive at an
operational proof that the success probability in
\eqref{eq:succ-prob-coh} does not increase under the action of a
superchannel. To see this, consider that a particular strategy of the prover
for coherently distinguishing the channels $\mathcal{N}^{0}$ and
$\mathcal{N}^{1}$ is to prepare a state $|\psi\rangle_{RC}$ and act with an
isometric extension $U_{C\rightarrow AMM_{1}}^{\mathcal{E}}$ of the
pre-processing $\mathcal{E}_{C\rightarrow AM}$ on system $C$. Then the
verifier performs the controlled unitary in \eqref{eq:controlled-unitary}, and
the prover performs an isometric extension $U_{BM\rightarrow DM_{2}%
}^{\mathcal{D}}$ of the post-processing $\mathcal{D}_{BM\rightarrow D}$, the
unitary $V_{RBE^{\prime}F}$, and the adjoint of $U_{BM\rightarrow DM_{2}%
}^{\mathcal{D}}$ (the last being implemented by a unitary and a projection).
The verifier finally performs the inverse of \eqref{eq:controlled-unitary} and
the projective measurement in \eqref{eq:verifier-final-proj-meas}. Since the
success probability does not increase under the action of the adjoint of
$U_{C\rightarrow AMM_{1}}^{\mathcal{E}}$ and since this is a particular
strategy for coherent discrimination of $\mathcal{N}^{0}$ and $\mathcal{N}%
^{1}$, while being a general strategy for coherent discrimination of
$\Theta(\mathcal{N}^{0})$ and $\Theta(\mathcal{N}^{1})$, we conclude that the
success probability does not increase under the action of a superchannel:

\begin{theorem}
\label{thm:monotone-super-ch}Let $\mathcal{N}_{A\rightarrow B}^{0}$ and
$\mathcal{N}_{A\rightarrow B}^{1}$ be quantum channels, and let $\Theta
_{\left(  A\rightarrow B\right)  \rightarrow(C\rightarrow D)}$ be a quantum
superchannel. Then the success probability of coherent channel discrimination
in \eqref{eq:succ-prob-coh}\ does not increase under the action of
$\Theta_{\left(  A\rightarrow B\right)  \rightarrow(C\rightarrow D)}$:%
\begin{equation}
p_{s}^{\text{coh}}(\mathcal{N}^{0},\mathcal{N}^{1})\geq p_{s}^{\text{coh}%
}(\Theta(\mathcal{N}^{0}),\Theta(\mathcal{N}^{1})).
\label{eq:non-increase-superch}%
\end{equation}

\end{theorem}

A strictly mathematical proof of \eqref{eq:non-increase-superch} is to employ
\eqref{eq:coh-ch-disc-suc-prob-derive}, the fundamental theorem of
superchannels in \eqref{eq:FTOSCs}, and the fact that the $\infty$-norm does
not increase under the action of a completely positive unital map or a projection.

\subsection{Computable by semi-definite programming}

The success probability in \eqref{eq:succ-prob-coh} can be computed by means
of the following semi-definite program:%
\begin{equation}
\sup_{\substack{\sigma_{R_{1}FBE},\\\rho_{A}}}\left\{  \operatorname{Tr}%
[Y_{R_{1}FBE}\sigma_{R_{1}FBE}]:\operatorname{Tr}_{BF}[\sigma_{R_{1}%
FBE}]=Z_{R_{1}E}^{\rho}\right\}  , \label{eq:SDP-succ-prob-coh-disc}%
\end{equation}
where $\sigma_{R_{1}FBE}$ and $\rho_{A}$ are density operators and%
\begin{align*}
Y_{R_{1}FBE}  &  :=\frac{1}{2}\sum_{i,j\in\left\{  0,1\right\}  ,k,\ell
}|ii\rangle\langle jj|_{R_{1}F}\otimes N_{k}^{i}N_{\ell}^{j\dag}%
\otimes|k\rangle\langle\ell|_{E},\\
Z_{R_{1}E}^{\rho}  &  :=\frac{1}{2}\sum_{i,j\in\left\{  0,1\right\}  ,k,\ell
}\operatorname{Tr}[N_{\ell}^{j\dag}N_{k}^{i}\rho_{A}]|i\rangle\langle
j|_{R_{1}}\otimes|k\rangle\langle\ell|_{E},
\end{align*}
with $\{N_{k}^{i}\}_{k}$ a set of Kraus operators for the channel
$\mathcal{N}_{A\rightarrow B}^{i}$ for $i\in\left\{  0,1\right\}  $. This
follows from the observation that coherent channel discrimination is a quantum
interactive proof, and the acceptance probability of any quantum interactive
proof can be calculated by means of a semi-definite program \cite{KW00,VW16}.
In the above semi-definite program, the density operator $\rho_{A}$ can be
understood as the reduction of the initial state of the prover on system $A$,
and the density operator $Z_{R_{1}E}^{\rho}$ is the reduced state from
\eqref{eq:state-verifier-after-1st-unitary}\ on systems $R_{1}E$. The
projection $Y_{R_{1}FBE}$ corresponds to the concatenation of the inverse
unitary in \eqref{eq:inverse-unitary} followed by the projection in
\eqref{eq:verifier-final-proj-meas}\ onto the accepting subspace. The equality
constraint in \eqref{eq:SDP-succ-prob-coh-disc} corresponds to the fact that
the state of the verifier on systems $R_{1}$ and $E$ should be the same before
and after the prover acts with the unitary $V_{RBE^{\prime}F}$.

The dual semi-definite program is given by%
\begin{align}
\inf_{\lambda,W_{R_{1}E}}\lambda &  \qquad\text{subject to}%
\label{eq:dual-SDP-succ-prob-coh}\\
W_{R_{1}E}\otimes I_{BF}  &  \geq Y_{R_{1}EBF},\\
\lambda I_{A}  &  \geq\frac{1}{2}\sum_{i,j\in\left\{  0,1\right\}  ,k,\ell
}w_{R_{1}E}^{ikj\ell}N_{k}^{i\dag}N_{\ell}^{j},
\end{align}
where $\lambda\in\mathbb{R}$, the operator $W_{R_{1}E}$ is Hermitian, and
$w_{R_{1}E}^{ikj\ell}:=\langle i|_{R_{1}}\langle k|_{E}W_{R_{1}E}%
|j\rangle_{R_{1}}|\ell\rangle_{E}$. This follows by the standard Lagrange
multiplier method.

\section{Incoherent Channel Discrimination with Uncomputing}

Another variation of channel discrimination is to follow the same protocol for
coherent channel discrimination but have the initial state be the maximally
mixed state and the final measurement be as in
\eqref{eq:final-meas-incoh-ch-disc}, with the first outcome indicating
success. So this is the main difference with coherent channel discrimination,
and the main difference with incoherent channel discrimination is that we
include a step for uncomputing. Let $p_{s}^{\text{inc,unc}}(\mathcal{N}%
^{0},\mathcal{N}^{1})$ denote the success probability for this case. We then
have the following bounds, implying \eqref{eq:coh-inc-bounds}:%
\begin{equation}
p_{s}^{\text{coh}}\leq p_{s}^{\text{inc,unc}}\leq p_{s}^{\text{inc}}\leq
\sqrt{p_{s}^{\text{coh}}}, \label{eq:bounds-coh-inc-incunc}%
\end{equation}
where the channel arguments are left implicit for brevity.

\section{Conclusion}

This paper has introduced a coherent version of quantum channel discrimination
and investigated various aspects of the success probability. I have proven an
alternate expression for it in
Proposition~\ref{prop:coh-ch-disc-suc-prob-derive}, some bounds in
Proposition~\ref{prop:bounds-suc-prob-1/2-1}\ and
Eq.~\eqref{eq:bounds-coh-inc-incunc}, that it does not increase under the
action of a quantum superchannel, and that it can be calculated by means of a
semi-definite program. An intriguing open question is to determine if
$p_{s}^{\text{coh}}-1/2$ is a metric on quantum channels. Consider that
$p_{s}^{\text{inc}}-1/2$ is, as is clear from~\eqref{eq:incoh-succ-prob-ch-disc}.

\section*{Acknowledgment}

I thank Stefan B\"{a}uml, Siddhartha Das, Felix Leditzky, and Xin Wang for
discussions related to the topic of this paper. I also acknowledge support
from the National Science Foundation under grant no.~1907615.

\newpage
\bibliographystyle{IEEEtran}
\bibliography{Ref}

\newpage
\appendices

\section{Proof of Proposition~\ref{prop:coh-ch-disc-suc-prob-derive}}

\label{app:proof-of-prop-succ-prob-coh}Let us begin with the expression in
\eqref{eq:unopt-succ-prob} for the unoptimized success probability:%
\begin{align}
&  \frac{1}{2}\left\Vert \langle\Phi|_{R_{1}F}\sum_{i\in\left\{  0,1\right\}
}|i\rangle_{R_{1}}U^{i\dag}VU^{i}|\psi\rangle_{RA}|00\rangle_{E^{\prime}%
F}\right\Vert _{2}^{2}\nonumber\\
&  =\frac{1}{4}\left\Vert \sum_{i,j\in\left\{  0,1\right\}  }\langle
jj|_{R_{1}F}|i\rangle_{R_{1}}U^{i\dag}VU^{i}|\psi\rangle_{RA}|00\rangle
_{E^{\prime}F}\right\Vert _{2}^{2}\\
&  =\frac{1}{4}\left\Vert \sum_{i\in\left\{  0,1\right\}  }\langle
i|_{F}U^{i\dag}VU^{i}|\psi\rangle_{RA}|00\rangle_{E^{\prime}F}\right\Vert
_{2}^{2}\\
&  =\frac{1}{4}\left\Vert \sum_{i\in\left\{  0,1\right\}  }\langle
i|_{F}\mathcal{N}^{i\dag}(V)|\psi\rangle_{RA}|00\rangle_{E^{\prime}%
F}\right\Vert _{2}^{2}, \label{eq:pf-simpler-exp-1}%
\end{align}
where we used the fact that $U^{i\dag}VU^{i}$ can be expressed in terms of the
channel adjoint as $\mathcal{N}^{i\dag}(V)$ \cite{W17book}. Let us write%
\begin{equation}
V_{RBE^{\prime}F}=\sum_{j,k\in\left\{  0,1\right\}  }Q_{RBE^{\prime}}%
^{j,k}\otimes|j\rangle\langle k|_{F}, \label{eq:V-decomp}%
\end{equation}
where the operators $Q_{RBE^{\prime}}^{j,k}$ satisfy%
\begin{align}
Q_{RBE^{\prime}}^{0,0\dag}Q_{RBE^{\prime}}^{0,0}+Q_{RBE^{\prime}}^{1,0\dag
}Q_{RBE^{\prime}}^{1,0}  &  =I_{RBE^{\prime}},\\
Q_{RBE^{\prime}}^{0,1\dag}Q_{RBE^{\prime}}^{0,0}+Q_{RBE^{\prime}}^{1,1\dag
}Q_{RBE^{\prime}}^{1,0}  &  =0,\\
Q_{RBE^{\prime}}^{0,0\dag}Q_{RBE^{\prime}}^{0,1}+Q_{RBE^{\prime}}^{1,0\dag
}Q_{RBE^{\prime}}^{1,1}  &  =0,\\
Q_{RBE^{\prime}}^{0,1\dag}Q_{RBE^{\prime}}^{0,1}+Q_{RBE^{\prime}}^{1,1\dag
}Q_{RBE^{\prime}}^{1,1}  &  =I_{RBE^{\prime}},
\end{align}
in order for $V_{RBE^{\prime}F}$ to be unitary. Then we find that%
\begin{align}
&  \langle i|_{F}\mathcal{N}^{i\dag}(V)|\psi\rangle_{RA}|00\rangle_{E^{\prime
}F}\nonumber\\
&  =\sum_{j,k\in\left\{  0,1\right\}  }\langle i|_{F}\left[  \mathcal{N}%
^{i\dag}(Q_{RBE^{\prime}}^{j,k})\otimes|j\rangle\langle k|_{F}\right]
|\psi\rangle_{RA}|00\rangle_{E^{\prime}F}\nonumber\\
&  =\mathcal{N}^{i\dag}(Q_{RBE^{\prime}}^{i,0})|\psi\rangle_{RA}%
|0\rangle_{E^{\prime}}\nonumber\\
&  =\mathcal{N}^{i\dag}(Q_{RBE^{\prime}}^{i,0}|0\rangle_{E^{\prime}}%
)|\psi\rangle_{RA} \label{eq:helper-calc-overlap-Q}%
\end{align}
which leads to%
\begin{equation}
\text{Eq.}~\eqref{eq:pf-simpler-exp-1}=\frac{1}{4}\left\Vert \sum
_{i\in\left\{  0,1\right\}  }\mathcal{N}^{i\dag}(Q_{RBE^{\prime}}%
^{i,0}|0\rangle_{E^{\prime}})|\psi\rangle_{RA}\right\Vert _{2}^{2}.
\end{equation}
Now optimizing over all input states $|\psi\rangle_{RA}$ and unitaries
$V_{RBE^{\prime}F}$, while setting%
\begin{equation}
P_{RB\rightarrow RBE^{\prime}}^{i}:=Q_{RBE^{\prime}}^{i,0}|0\rangle
_{E^{\prime}}, \label{eq:Pi-def}%
\end{equation}
we find that%
\begin{multline}
p_{s}^{\text{coh}}(\mathcal{N}^{0},\mathcal{N}^{1})=\\
\sup_{_{\substack{\left\{  P^{i}\right\}  _{i\in\left\{  0,1\right\}  }%
:\\\sum_{i}P^{i\dag}P^{i}=I_{RB}}}}\frac{1}{4}\left\Vert \sum_{i\in\left\{
0,1\right\}  }\mathcal{N}^{i\dag}(P_{RB\rightarrow RBE^{\prime}}%
^{i})\right\Vert _{\infty}^{2}.
\end{multline}
Also, note that to any set $\left\{  P^{i}\right\}  _{i\in\left\{
0,1\right\}  }$ satisfying $\sum_{i}P^{i\dag}P^{i}=I_{RB}$, we can complete it
to a unitary $V_{RBE^{\prime}F}$.

Since the unitary $V_{RBE^{\prime}F}$ implements a quantum channel from
systems $RB$ to $BF$, and since the dimension of the environment of any
quantum channel need not be larger than the product of the input and output
dimensions, it suffices to take $\left\vert E^{\prime}\right\vert =\left\vert
R\right\vert \left\vert B\right\vert ^{2}\left\vert F\right\vert $. Since
$\left\vert R\right\vert =\left\vert A\right\vert $ and $\left\vert
F\right\vert =2$, it suffices to take $\left\vert E^{\prime}\right\vert
=2\left\vert A\right\vert \left\vert B\right\vert ^{2}$ as claimed.

\section{Proof of Proposition~\ref{prop:bounds-suc-prob-1/2-1}}

The upper bound in \eqref{eq:prob-suc-coh-disc-bnds}\ trivially follows
because $p_{s}^{\text{coh}}(\mathcal{N}^{0},\mathcal{N}^{1})$ is a
probability. The lower bound in \eqref{eq:prob-suc-coh-disc-bnds}\ follows by
picking $P_{RB\rightarrow RBE^{\prime}}^{i}=\sqrt{1/2}I_{RB}\otimes
|0\rangle_{E^{\prime}}$ for $i\in\left\{  0,1\right\}  $ and evaluating
\eqref{eq:coh-ch-disc-suc-prob-derive}. Consider that%
\begin{align}
&  \left[  \frac{1}{2}\sup_{\substack{\left\{  P^{i}\right\}  _{i\in\left\{
0,1\right\}  }:\\\sum_{i}P^{i\dag}P^{i}=I_{RB}}}\left\Vert \sum_{i\in\left\{
0,1\right\}  }(\mathcal{N}_{A\rightarrow B}^{i})^{\dag}(P_{RB\rightarrow
RBE^{\prime}}^{i})\right\Vert _{\infty}\right]  ^{2}\nonumber\\
&  \geq\left[  \frac{1}{2}\left\Vert \sqrt{\frac{1}{2}}\sum_{i\in\left\{
0,1\right\}  }(\mathcal{N}_{A\rightarrow B}^{i})^{\dag}(I_{RB}\otimes
|0\rangle_{E^{\prime}})\right\Vert _{\infty}\right]  ^{2}\\
&  =\left[  \frac{1}{2\sqrt{2}}\left\Vert 2(I_{RB}\otimes|0\rangle_{E^{\prime
}})\right\Vert _{\infty}\right]  ^{2}\\
&  =\left[  \frac{1}{\sqrt{2}}\right]  ^{2}=\frac{1}{2}.
\end{align}
The first inequality follows by picking $P_{RB\rightarrow RBE^{\prime}}^{i}$
as indicated. The first equality follows because $(\mathcal{N}_{A\rightarrow
B}^{i})^{\dag}$ is a unital map.

If the channels $\mathcal{N}_{A\rightarrow B}^{0}$ and $\mathcal{N}%
_{A\rightarrow B}^{1}$ are the same (so that $\mathcal{N}_{A\rightarrow B}%
^{0}=\mathcal{N}_{A\rightarrow B}^{1}= \mathcal{N}_{A\rightarrow B}$), then
consider for $\left\{  P^{i}\right\}  _{i\in\left\{  0,1\right\}  }$
satisfying $\sum_{i\in\left\{  0,1\right\}  }P^{i\dag}P^{i}=I_{RB}$ that%
\begin{align}
&  \left[  \frac{1}{2}\left\Vert \sum_{i\in\left\{  0,1\right\}  }%
(\mathcal{N}_{A\rightarrow B}^{i})^{\dag}(P_{RB\rightarrow RBE^{\prime}}%
^{i})\right\Vert _{\infty}\right]  ^{2}\nonumber\\
&  =\left[  \frac{1}{2}\left\Vert (\mathcal{N}_{A\rightarrow B})^{\dag}\left(
\sum_{i\in\left\{  0,1\right\}  }P_{RB\rightarrow RBE^{\prime}}^{i}\right)
\right\Vert _{\infty}\right]  ^{2}\\
&  \leq\left[  \frac{1}{2}\left\Vert \sum_{i\in\left\{  0,1\right\}
}P_{RB\rightarrow RBE^{\prime}}^{i}\right\Vert _{\infty}\right]  ^{2}\\
&  =\frac{1}{4}\left\Vert \sum_{i\in\left\{  0,1\right\}  }P_{RB\rightarrow
RBE^{\prime}}^{i}\right\Vert _{\infty}^{2}\\
&  =\frac{1}{4}\left\Vert \left(  \sum_{j\in\left\{  0,1\right\}
}P_{RB\rightarrow RBE^{\prime}}^{j\dag}\right)  \left(  \sum_{i\in\left\{
0,1\right\}  }P_{RB\rightarrow RBE^{\prime}}^{i}\right)  \right\Vert _{\infty
}\\
&  =\frac{1}{4}\left\Vert P^{0\dag}P^{0}+P^{1\dag}P^{1}+P^{1\dag}%
P^{0}+P^{0\dag}P^{1}\right\Vert _{\infty}\\
&  \leq\frac{1}{4}\left\Vert 2\left(  P^{0\dag}P^{0}+P^{1\dag}P^{1}\right)
\right\Vert _{\infty}\\
&  =1/2.
\end{align}
The first equality follows from the assumption that the channels are the same.
The first inequality follows because the operator norm is non-increasing under
the action of a completely positive unital map \cite{P03}. The third equality
follows because $\left\Vert A\right\Vert _{\infty}^{2}=\left\Vert A^{\dag
}A\right\Vert _{\infty}$. The second inequality follows because%
\begin{equation}
P^{1\dag}P^{0}+P^{0\dag}P^{1}\leq P^{0\dag}P^{0}+P^{1\dag}P^{1},
\label{eq:proj-ineq-basic}%
\end{equation}
which is equivalent to%
\begin{equation}
\left(  P^{0\dag}-P^{1\dag}\right)  \left(  P^{0}-P^{1}\right)  \geq0.
\end{equation}
The final equality follows because $\sum_{i\in\left\{  0,1\right\}  }P^{i\dag
}P^{i}=I_{RB}$ and $\left\Vert I\right\Vert _{\infty}=1$. Since the lower
bound in \eqref{eq:prob-suc-coh-disc-bnds} always holds, we conclude that
$p_{s}^{\text{coh}}(\mathcal{N}^{0},\mathcal{N}^{1})=1/2$ if $\mathcal{N}%
^{0}=\mathcal{N}^{1}$.

If the channels $\mathcal{N}^{0}$ and $\mathcal{N}^{1}$ are perfectly
distinguishable, then this means that $p_{s}^{\text{inc}}(\mathcal{N}%
^{0},\mathcal{N}^{1})=1$. Applying the upper bound in
\eqref{eq:coh-inc-bounds} implies that $p_{s}^{\text{coh}}(\mathcal{N}%
^{0},\mathcal{N}^{1})=1$. If instead $p_{s}^{\text{coh}}(\mathcal{N}%
^{0},\mathcal{N}^{1})=1$, then the lower bound in \eqref{eq:coh-inc-bounds}
implies that $p_{s}^{\text{inc}}(\mathcal{N}^{0},\mathcal{N}^{1})=1$. Then if
$p_{s}^{\text{inc}}(\mathcal{N}^{0},\mathcal{N}^{1})=1$, it is known that
$\mathcal{N}^{0}$ and $\mathcal{N}^{1}$ are perfectly distinguishable. The
bounds in \eqref{eq:coh-inc-bounds} are proved in
Appendix~\ref{app:bnds-proofs-coh-inc-incunc}.

\section{Proof of Theorem~\ref{thm:monotone-super-ch}}

This appendix establishes a proof of Theorem~\ref{thm:monotone-super-ch},
which states that the success probability in \eqref{eq:succ-prob-coh} does not
increase under the action of a quantum superchannel.

First let us consider a more mathematical proof that requires fewer steps. Let
$\Theta_{\left(  A\rightarrow B\right)  \rightarrow\left(  C\rightarrow
D\right)  }$ denote a quantum superchannel. Exploiting the expression in
\eqref{eq:coh-ch-disc-suc-prob-derive}, we can write the success probability
$p_{s}^{\text{coh}}(\Theta(\mathcal{N}^{0}),\Theta(\mathcal{N}^{1}))$ as%
\begin{multline}
p_{s}^{\text{coh}}(\Theta(\mathcal{N}^{0}),\Theta(\mathcal{N}^{1}))=\\
\left[  \frac{1}{2}\sup_{\substack{\left\{  Q^{i}\right\}  _{i}:\\\sum
_{i}Q^{i\dag}Q^{i}=I}}\left\Vert \sum_{i\in\left\{  0,1\right\}  }%
(\Theta(\mathcal{N}_{A\rightarrow B}^{i}))^{\dag}(Q_{RD\rightarrow
RDE^{\prime\prime}}^{i})\right\Vert _{\infty}\right]  ^{2}.
\end{multline}
Let $\left\{  Q^{i}\right\}  _{i\in\left\{  0,1\right\}  }$ be an arbitrary
set of operators satisfying $\sum_{i}Q^{i\dag}Q^{i}=I_{RD}$. Then consider from the fundamental theorem of
superchannels in \eqref{eq:FTOSCs}
that%
\begin{align}
&  \left\Vert \sum_{i\in\left\{  0,1\right\}  }(\Theta(\mathcal{N}%
_{A\rightarrow B}^{i}))^{\dag}(Q_{RD\rightarrow RDE^{\prime\prime}}%
^{i})\right\Vert _{\infty}\nonumber\\
&  =\left\Vert \sum_{i\in\left\{  0,1\right\}  }((\mathcal{E}_{C\rightarrow
AM})^{\dag}\circ(\mathcal{N}_{A\rightarrow B}^{i})^{\dag}\circ(\mathcal{D}%
_{BM\rightarrow D})^{\dag})(Q^{i})\right\Vert _{\infty}\nonumber\\
&  \leq\left\Vert \sum_{i\in\left\{  0,1\right\}  }((\mathcal{N}_{A\rightarrow
B}^{i})^{\dag}\circ(\mathcal{D}_{BM\rightarrow D})^{\dag})(Q_{RD\rightarrow
RDE^{\prime\prime}}^{i})\right\Vert _{\infty},\label{eq:superch-mathy-1}%
\end{align}
where the inequality follows from the fact that a completely positive unital
map does not increase the operator norm \cite{P03}. Let $W_{BMG^{\prime
}\rightarrow DM_{2}}$ be a unitary extension of $\mathcal{D}_{BM\rightarrow
D}$, with input environment $G^{\prime}$ and output environment $M_{2}$, so
that%
\begin{multline}
\mathcal{D}_{BM\rightarrow D}(\cdot)=\\
\operatorname{Tr}_{M_{2}}[W_{BMG^{\prime}\rightarrow DM_{2}}[(\cdot
)\otimes|0\rangle\langle0|_{G^{\prime}}](W_{BMG^{\prime}\rightarrow DM_{2}%
})^{\dag}].
\end{multline}
Then it is well known \cite{W17book} that%
\begin{multline}
(\mathcal{D}_{BM\rightarrow D})^{\dag}(\cdot)=\\
\langle0|_{G^{\prime}}(W_{BMG^{\prime}\rightarrow DM_{2}})^{\dag}%
[(\cdot)\otimes I_{M_{2}}]W_{BMG^{\prime}\rightarrow DM_{2}}|0\rangle
_{G^{\prime}}.
\end{multline}
Substituting above, we find that%
\begin{align}
&  \text{Eq.}~\eqref{eq:superch-mathy-1}\nonumber\\
&  =\left\Vert \sum_{i\in\left\{  0,1\right\}  }(\mathcal{N}_{A\rightarrow
B}^{i})^{\dag}(\langle0|_{G^{\prime}}W^{\dag}(Q^{i}\otimes I_{M^{\prime}%
})W|0\rangle_{G^{\prime}})\right\Vert _{\infty}\nonumber\\
&  \leq\left\Vert \sum_{i\in\left\{  0,1\right\}  }(\mathcal{N}_{A\rightarrow
B}^{i})^{\dag}(W^{\dag}(Q_{RD\rightarrow RDE^{\prime\prime}}^{i}\otimes
I_{M^{\prime}})W|0\rangle_{G^{\prime}})\right\Vert _{\infty}\nonumber\\
&  \leq\sup_{\substack{\left\{  P^{i}\right\}  _{i\in\left\{  0,1\right\}
}:\\\sum_{i}P^{i\dag}P^{i}=I_{RB}}}\left\Vert \sum_{i\in\left\{  0,1\right\}
}(\mathcal{N}_{A\rightarrow B}^{i})^{\dag}(P_{RB\rightarrow RBE^{\prime}}%
^{i})\right\Vert _{\infty}\nonumber\\
&  =2 \sqrt{p_{s}^{\text{coh}}(\mathcal{N}^{0},\mathcal{N}^{1})}.
\end{align}
The first inequality follows because a projection onto $\langle0|_{G^{\prime}%
}$ does not increase the operator norm, and the second inequality follows
because the set $\{W^{\dag}(Q_{RD\rightarrow RDE^{\prime\prime}}^{i}\otimes
I_{M^{\prime}})W|0\rangle_{G^{\prime}}\}_{i}$ is a particular choice of
$\left\{  P^{i}\right\}  _{i\in\left\{  0,1\right\}  }$ satisfying $\sum
_{i}P^{i\dag}P^{i}=I_{RB}$. We then conclude the inequality in \eqref{eq:non-increase-superch} since $\left\{  Q^{i}\right\}  _{i\in\left\{  0,1\right\}  }$ is an arbitrary
set of operators satisfying $\sum_{i}Q^{i\dag}Q^{i}=I_{RD}$.

\begin{figure*}[ptb]
\begin{center}
\includegraphics[
width=6.5518in
]{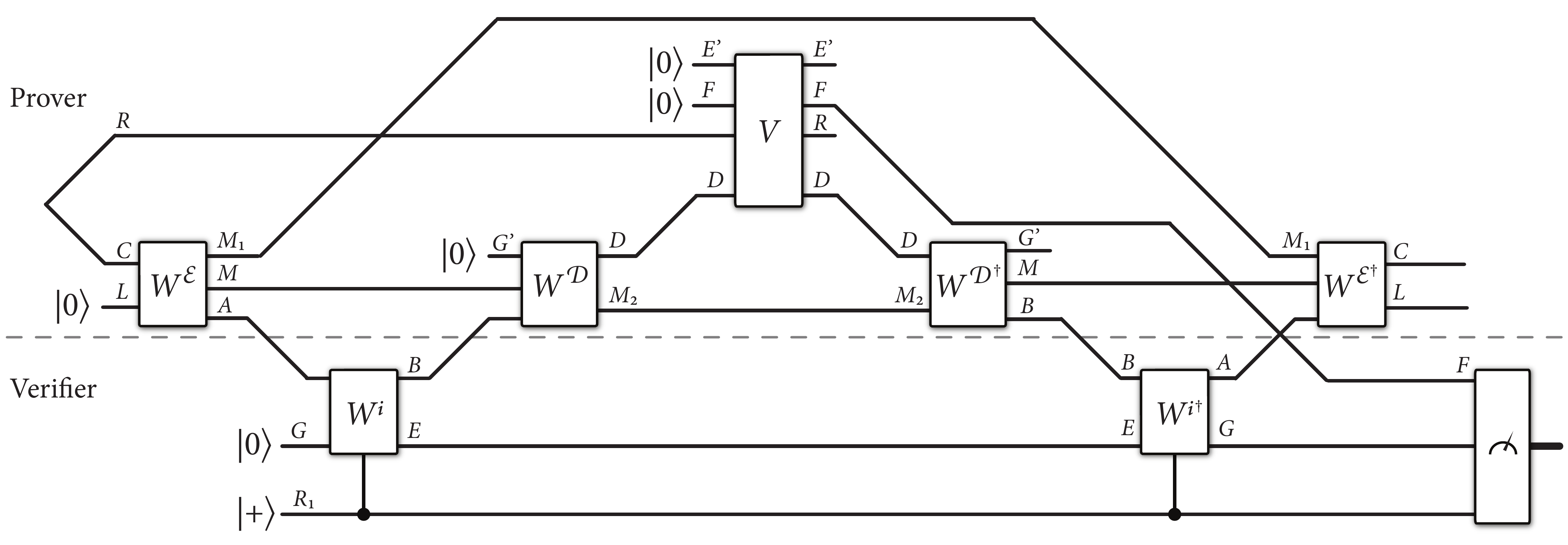}
\end{center}
\caption{Depiction of the operational proof of
Theorem~\ref{thm:monotone-super-ch}. This is a particular strategy for
coherent channel discrimination of the channels $\mathcal{N}_{A\rightarrow
B}^{0}$ and $\mathcal{N}_{A\rightarrow B}^{1}$, but it is a general strategy
for coherent channel discrimination of $\Theta(\mathcal{N}_{A\rightarrow
B}^{0})$ and $\Theta(\mathcal{N}_{A\rightarrow B}^{1})$, where $\Theta
_{\left(  A\rightarrow B\right)  \rightarrow\left(  C\rightarrow D\right)  }$
is a quantum superchannel.}
\label{fig:nonincrease-super-ch}
\end{figure*}

A more operational proof of Theorem~\ref{thm:monotone-super-ch} goes along the
lines discussed in Section~\ref{sec:nonincrease-superch} and is depicted in
Figure~\ref{fig:nonincrease-super-ch}. The main idea behind this operational
proof is that the strategy depicted in Figure~\ref{fig:nonincrease-super-ch}
is a particular strategy for coherent channel discrimination of the channels
$\mathcal{N}_{A\rightarrow B}^{0}$ and $\mathcal{N}_{A\rightarrow B}^{1}$, but
it is a general strategy for coherent channel discrimination of the channels
$\Theta(\mathcal{N}_{A\rightarrow B}^{0})$ and $\Theta(\mathcal{N}%
_{A\rightarrow B}^{1})$, where $\Theta_{\left(  A\rightarrow B\right)
\rightarrow\left(  C\rightarrow D\right)  }$ is a quantum superchannel. Let
$W_{AG\rightarrow BE}^{i}$ be a unitary extension of the channel
$\mathcal{N}_{A\rightarrow B}^{i}$, so that%
\begin{multline}
\mathcal{N}_{A\rightarrow B}^{i}(\cdot)=\\
\operatorname{Tr}_{E}[W_{AG\rightarrow BE}^{i}((\cdot)\otimes|0\rangle
\langle0|_{E})(W_{AG\rightarrow BE}^{i})^{\dag}].
\end{multline}
Let $\Theta_{\left(  A\rightarrow B\right)  \rightarrow\left(  C\rightarrow
D\right)  }$ be a quantum superchannel, and by the fundamental theorem of
superchannels, it has a physical realization as in \eqref{eq:FTOSCs}. Let
$W_{CL\rightarrow AMM_{1}}^{\mathcal{E}}$ be a unitary extension of the
channel $\mathcal{E}_{A\rightarrow BM}$, and let $W_{BMG^{\prime}\rightarrow
DM_{2}}^{\mathcal{D}}$ be a unitary extension of the channel $\mathcal{D}%
_{BM\rightarrow D}$, so that%
\begin{multline}
\mathcal{E}_{A\rightarrow BM}(\cdot)=\\
\operatorname{Tr}_{M_{1}}[W_{CL\rightarrow AMM_{1}}^{\mathcal{E}}%
((\cdot)\otimes|0\rangle\langle0|_{L})(W_{CL\rightarrow AMM_{1}}^{\mathcal{E}%
})^{\dag}].
\end{multline}
and%
\begin{multline}
\mathcal{D}_{BM\rightarrow D}(\cdot)=\\
\operatorname{Tr}_{M_{2}}[W_{BMG^{\prime}\rightarrow DM_{2}}^{\mathcal{D}%
}((\cdot)\otimes|0\rangle\langle0|_{G^{\prime}})(W_{BMG^{\prime}\rightarrow
DM_{2}}^{\mathcal{D}})^{\dag}].
\end{multline}
From \eqref{eq:pf-simpler-exp-1}, we know that the success probability for
coherent channel discrimination of $\Theta(\mathcal{N}_{A\rightarrow B}^{0})$
and $\Theta(\mathcal{N}_{A\rightarrow B}^{1})$, for a fixed strategy $\left\{
|\psi\rangle_{RC},V_{RDFE^{\prime}}\right\}  $, is given by%
\begin{equation}
\frac{1}{4}\left\Vert \sum_{i\in\left\{  0,1\right\}  }\langle i|_{F}T^{i\dag
}VT^{i}|\psi\rangle_{RC}|00\rangle_{E^{\prime}F}\right\Vert _{2}%
^{2},\label{eq:succ-prob-coh-with-superch-1}%
\end{equation}
where for $i\in\left\{  0,1\right\}  $,%
\begin{multline}
T_{C\rightarrow DM_{1}M_{2}E}^{i}:=
\\
W_{BMG^{\prime}\rightarrow DM_{2}%
}^{\mathcal{D}}W_{AG\rightarrow BE}^{i}W_{CL\rightarrow AMM_{1}}^{\mathcal{E}%
}|0\rangle_{L}|0\rangle_{G}|0\rangle_{G^{\prime}}.
\end{multline}
This follows because $T_{C\rightarrow DM_{1}M_{2}E}^{i}$ is an isometric
extension of the channel $\Theta(\mathcal{N}_{A\rightarrow B}^{i})$ for
$i\in\left\{  0,1\right\}  $. Since the Euclidean norm is non-increasing with
respect to the projections onto $|0\rangle_{L}$ and $|0\rangle_{G^{\prime}}$,
we find that%
\begin{align}
& \text{Eq.~\eqref{eq:succ-prob-coh-with-superch-1}}\notag \\
& \leq\frac{1}{4}\left\Vert \sum_{i\in\left\{  0,1\right\}  }\langle
i|_{F}\langle0|_{G}W^{\mathcal{E\dag}}W^{i\dag}W^{\mathcal{D\dag}}VT^{i}%
|\psi\rangle_{RC}|00\rangle_{E^{\prime}F}\right\Vert _{2}^{2} \notag \\
& =\frac{1}{4}\left\Vert \sum_{i\in\left\{  0,1\right\}  }\langle
i|_{F}\langle0|_{G}W^{i\dag}W^{\mathcal{D\dag}}VT^{i}|\psi\rangle
_{RC}|00\rangle_{E^{\prime}F}\right\Vert _{2}^{2}%
,\label{eq:super-ch-proof-oper-final}%
\end{align}
where for the last equality we used that $W_{CL\rightarrow AMM_{1}%
}^{\mathcal{E}}$ is unitary and not acting on any of the systems being
projected out. Now we observe that $W_{CL\rightarrow AMM_{1}}^{\mathcal{E}%
}|\psi\rangle_{RC}|0\rangle_{L}|0\rangle_{G^{\prime}}$ is a particular pure
state that the prover can use for coherent channel discrimination of
$\mathcal{N}_{A\rightarrow B}^{0}$ and $\mathcal{N}_{A\rightarrow B}^{1}$, and
$(W_{BMG^{\prime}\rightarrow DM_{2}}^{\mathcal{D}})^{\dag}V_{RDFE^{\prime}%
}W_{BMG^{\prime}\rightarrow DM_{2}}^{\mathcal{D}}$ is a particular unitary
that the prover can use for the same purpose. So we conclude that%
\begin{equation}
\text{Eq.}~\eqref{eq:super-ch-proof-oper-final}\leq p_{s}^{\text{coh}%
}(\mathcal{N}^{0},\mathcal{N}^{1}).
\end{equation}
Since the strategy employed for distinguishing $\Theta(\mathcal{N}%
_{A\rightarrow B}^{0})$ and $\Theta(\mathcal{N}_{A\rightarrow B}^{1})$ is
arbitrary, we conclude the operational proof of
Theorem~\ref{thm:monotone-super-ch}.

\section{Proof of Semi-definite programming formulation in
Eq.~\eqref{eq:SDP-succ-prob-coh-disc}}

Here I establish the particular form of the success probability in
\eqref{eq:SDP-succ-prob-coh-disc}, which demonstrates that
\eqref{eq:succ-prob-coh} can be calculated by means of a semi-definite
program. As stated previously, this follows from the fact that coherent
channel discrimination is a quantum interactive proof system, and the
acceptance probability of any quantum interactive proof system can be
calculated via a semi-definite program \cite{KW00,VW16}. In particular,
coherent channel discrimination is a three-message quantum interactive proof
system. Recall from \cite[Section~4.3]{VW16} that a three-message interactive
proof system is specified by two linear isometries $T_{Y_{0}\rightarrow
Z_{1}X_{1}}^{1}$ and $T_{Z_{1}Y_{1}\rightarrow Z_{2}}^{2}$ for the initial
circuit of the verifier and the final circuit of the verifier before the
measurement, respectively (in particular, see \cite[Figure~4.5]{VW16}). Then
the semi-definite program for the acceptance probability is given by
\cite[Figure~4.6]{VW16} as%
\begin{equation}
\sup\operatorname{Tr}[T^{2\dag}\Pi T^{2}\sigma_{1}],
\end{equation}
subject to%
\begin{align}
\operatorname{Tr}[\sigma_{0}]  &  =1,\\
\operatorname{Tr}_{Y_{1}}[\sigma_{1}]  &  =\operatorname{Tr}_{X_{1}}%
[T^{1}\sigma_{0}T^{1\dag}],
\end{align}
where $\Pi$ is the projection onto the accepting subspace and $\sigma_{i}$ is
positive semi-definite for $i\in\left\{  0,1\right\}  $. The constraints
correspond to the fact that the initial reduced state $\sigma_{0}$ of the
prover should be a density operator and that the reduced state of the verifier
on system $Z_{1}$ should be the same before and after the prover acts.

In our case, the initial reduced state $\sigma_{0}$ is on system $A$ and so we
can call it $\rho_{A}$. The isometry $T_{Y_{0}\rightarrow Z_{1}X_{1}}^{1}$
corresponds to the action in \eqref{eq:state-verifier-after-1st-unitary}:%
\begin{equation}
|\varphi\rangle_{A}\rightarrow\frac{1}{\sqrt{2}}\sum_{i\in\left\{
0,1\right\}  }|i\rangle_{R_{1}}U_{A\rightarrow BE}^{i}|\varphi\rangle_{A},
\end{equation}
where $U_{A\rightarrow BE}^{i}=W_{AG\rightarrow BE}^{i}|0\rangle_{G}$ and
system $Y_{0}$ is $A$, $Z_{1}$ is $R_{1}E$, and $X_{1}$ is $B$. The right-hand
side above can be rewritten using Kraus operators for channel $\mathcal{N}%
_{A\rightarrow B}^{i}$ as%
\begin{equation}
\frac{1}{\sqrt{2}}\sum_{i\in\left\{  0,1\right\}  ,k}|i\rangle_{R_{1}}%
N_{k}^{i}|\varphi\rangle_{A}|k\rangle_{E}.
\end{equation}
So we find that%
\begin{align}
&  \operatorname{Tr}_{X_{1}}[T^{1}\sigma_{0}T^{1\dag}]\nonumber\\
&  =\frac{1}{2}\operatorname{Tr}_{B}\left[  \sum_{i,j\in\left\{  0,1\right\}
,k,\ell}|i\rangle\langle j|_{R_{1}}\otimes N_{k}^{i}\rho_{A}N_{\ell}^{j\dag
}\otimes|k\rangle\langle\ell|_{E}\right] \nonumber\\
&  =\frac{1}{2}\sum_{i,j\in\left\{  0,1\right\}  ,k,\ell}\operatorname{Tr}%
[N_{k}^{i}\rho_{A}N_{\ell}^{j\dag}]|i\rangle\langle j|_{R_{1}}\otimes
|k\rangle\langle\ell|_{E}\nonumber\\
&  =\frac{1}{2}\sum_{i,j\in\left\{  0,1\right\}  ,k,\ell}\operatorname{Tr}%
[N_{\ell}^{j\dag}N_{k}^{i}\rho_{A}]|i\rangle\langle j|_{R_{1}}\otimes
|k\rangle\langle\ell|_{E}\nonumber\\
&  =Z_{R_{1}E}^{\rho},
\end{align}
as defined after \eqref{eq:SDP-succ-prob-coh-disc}. The state $\sigma_{1}$ of
the prover on systems $Z_{1}$ and $Y_{1}$ is denoted by $\sigma_{R_{1}EFB}$,
with $Z_{1}$ being $R_{1}E$ and $Y_{1}$ being $FB$. The isometry
$T_{Z_{1}Y_1 \rightarrow Z_{2}}^{2}$ corresponds to the action in
\eqref{eq:inverse-unitary} (the inverse controlled unitary), with $Z_{2}$ being $R_{1}AGF$. The projection $\Pi$ onto the
accepting subspace is $\Phi_{R_{1}F}\otimes I_{A}\otimes|0\rangle\langle
0|_{G}$. So we find that%
\begin{align}
&  T^{2\dag}\Pi T^{2}\nonumber\\
&  =\left(  \sum_{i\in\left\{  0,1\right\}  }|i\rangle\langle i|_{R_{1}%
}\otimes W_{AG\rightarrow BE}^{i}\right)  \Phi_{R_{1}F}\otimes I_{A}%
\otimes|0\rangle\langle0|_{G}\nonumber\\
&  \qquad\times\left(  \sum_{j\in\left\{  0,1\right\}  }|j\rangle\langle
j|_{R_{1}}\otimes(W_{AG\rightarrow BE}^{j})^{\dag}\right) \nonumber\\
&  =\sum_{i,j\in\left\{  0,1\right\}  }|i\rangle\langle i|_{R_{1}}\Phi
_{R_{1}F}|j\rangle\langle j|_{R_{1}}\otimes N_{k}^{i}N_{\ell}^{j\dag}%
\otimes|k\rangle\langle\ell|_{E}\nonumber\\
&  =\frac{1}{2}\sum_{i,j\in\left\{  0,1\right\}  }|ii\rangle\langle
jj|_{R_{1}F}\otimes N_{k}^{i}N_{\ell}^{j\dag}\otimes|k\rangle\langle\ell
|_{E}\nonumber\\
&  =Y_{R_{1}FBE},%
\end{align}
as defined after \eqref{eq:SDP-succ-prob-coh-disc}. This concludes the proof
of \eqref{eq:SDP-succ-prob-coh-disc}.

The dual program in \eqref{eq:dual-SDP-succ-prob-coh}\ follows from standard
techniques (Lagrange multiplier method) or by plugging into \cite[Figure~4.7]%
{VW16}.

\section{Proof of bounds relating success probabilities in coherent and
incoherent channel discrimination}

\label{app:bnds-proofs-coh-inc-incunc}

This appendix establishes a proof of the bounds in
\eqref{eq:bounds-coh-inc-incunc}. Let us begin by establishing an expression
for $p_{s}^{\text{inc,unc}}(\mathcal{N}^{0},\mathcal{N}^{1})$:

\begin{proposition}
Let $\mathcal{N}_{A\rightarrow B}^{0}$ and $\mathcal{N}_{A\rightarrow B}^{1}$
be quantum channels. Then the success probability of incoherent channel
discrimination with uncomputing can be written as%
\begin{multline}
p_{s}^{\text{inc,unc}}(\mathcal{N}^{0},\mathcal{N}^{1}%
)=\label{eq:succ-prob-incunc-exp}\\
\frac{1}{2}\sup_{_{\substack{\left\{  P^{i}\right\}  _{i\in\left\{
0,1\right\}  }:\\\sum_{i}P^{i\dag}P^{i}=I_{RB}}}}\left\Vert \sum_{i\in\left\{
0,1\right\}  }\mathcal{N}^{i\dag}(P^{i\dag})\mathcal{N}^{i\dag}(P^{i}%
)\right\Vert _{\infty},
\end{multline}
where the operators are written without abbreviation as $P_{RB\rightarrow
RBE^{\prime}}^{i}$.
\end{proposition}

\begin{proof}
The analysis is similar to that given in
Appendix~\ref{app:proof-of-prop-succ-prob-coh}. Considering that the initial
state of the verifier, the maximally mixed state, can be purified by the
maximally entangled state $|\Phi\rangle_{R_{2}R_{1}}$ and by running through a
calculation similar to that given in
\eqref{eq:controlled-unitary}--\eqref{eq:unopt-succ-prob}, we find that the
unoptimized success probability for a fixed strategy of the prover is equal to%
\begin{multline}
\frac{1}{2}\left\Vert \Pi_{R_{1}F}\sum_{i\in\left\{  0,1\right\}  }%
|i\rangle_{R_{2}}|i\rangle_{R_{1}}U^{i\dag}VU^{i}|\psi\rangle_{RA}%
|00\rangle_{E^{\prime}F}\right\Vert _{2}^{2}%
\label{eq:succ-prob-incunc-unopt-1}\\
=\frac{1}{2}\left\Vert \Pi_{R_{1}F}\sum_{i\in\left\{  0,1\right\}  }%
|i\rangle_{R_{2}}|i\rangle_{R_{1}}\mathcal{N}^{i\dag}(V)|\psi\rangle
_{RA}|00\rangle_{E^{\prime}F}\right\Vert _{2}^{2}%
\end{multline}
where%
\begin{equation}
\Pi_{R_{1}F}:=\sum_{j\in\left\{  0,1\right\}  }|jj\rangle\langle jj|_{R_{1}F}.
\end{equation}
This leads to%
\begin{align}
&  \text{Eq.}~\eqref{eq:succ-prob-incunc-unopt-1}\nonumber\\
&  =\frac{1}{2}\left\Vert \sum_{i\in\left\{  0,1\right\}  }|iii\rangle
_{R_{2}R_{1}F}\langle i|_{F}\mathcal{N}^{i\dag}(V)|\psi\rangle_{RA}%
|00\rangle_{E^{\prime}F}\right\Vert _{2}^{2}\\
&  =\frac{1}{2}\left\Vert \sum_{i\in\left\{  0,1\right\}  }|iii\rangle
_{R_{2}R_{1}F}\mathcal{N}^{i\dag}(P_{RB\rightarrow RBE^{\prime}}^{i}%
)|\psi\rangle_{RA}\right\Vert _{2}^{2}\\
&  =\frac{1}{2}\sum_{i\in\left\{  0,1\right\}  }\langle\psi|_{RA}%
\mathcal{N}^{i\dag}((P^{i})^{\dag})\mathcal{N}^{i\dag}(P^{i})|\psi\rangle
_{RA},
\end{align}
where in the second line we made use of \eqref{eq:helper-calc-overlap-Q} and
\eqref{eq:Pi-def}, and the last line follows by direct evaluation of the norm.
Now optimizing over all strategies of the prover and employing the definition
of the operator norm, we conclude \eqref{eq:succ-prob-incunc-exp}.
\end{proof}

\bigskip

We can now establish \eqref{eq:bounds-coh-inc-incunc}. Let us start by
proving%
\begin{equation}
p_{s}^{\text{coh}}(\mathcal{N}^{0},\mathcal{N}^{1})\leq p_{s}^{\text{inc,unc}%
}(\mathcal{N}^{0},\mathcal{N}^{1}). \label{eq:1st-ineq-bounds}%
\end{equation}
Starting from \eqref{eq:coh-ch-disc-suc-prob-derive}, let $\left\{
P_{RB\rightarrow RBE^{\prime}}^{i}\right\}  _{i\in\left\{  0,1\right\}  }$ be
arbitrary operators satisfying $\sum_{i}P^{i\dag}P^{i}=I_{RB}$. Then%
\begin{align}
&  \frac{1}{4}\left\Vert \sum_{i\in\left\{  0,1\right\}  }\mathcal{N}^{i\dag
}(P^{i})\right\Vert _{\infty}^{2}\\
&  =\frac{1}{4}\left\Vert \left(  \sum_{j\in\left\{  0,1\right\}  }%
\mathcal{N}^{j\dag}(P^{j\dag})\right)  \left(  \sum_{i\in\left\{  0,1\right\}
}\mathcal{N}^{i\dag}(P^{i})\right)  \right\Vert _{\infty}\\
&  \leq\frac{1}{2}\left\Vert \sum_{i\in\left\{  0,1\right\}  }\mathcal{N}%
^{i\dag}(P^{i\dag})  \mathcal{N}^{i\dag}(P^{i})  \right\Vert
_{\infty}\\
&  \leq p_{s}^{\text{inc,unc}}(\mathcal{N}^{0},\mathcal{N}^{1})
\end{align}
The equality follows because $\left\Vert X\right\Vert _{\infty}^{2}=\left\Vert
X^{\dag}X\right\Vert _{\infty}$ for any operator $X$. The first inequality
follows from the same reasoning as in \eqref{eq:proj-ineq-basic}. Since
$\left\{  P_{RB\rightarrow RBE^{\prime}}^{i}\right\}  _{i\in\left\{
0,1\right\}  }$ satisfying $\sum_{i}P^{i\dag}P^{i}=I_{RB}$ is arbitrary, we
conclude \eqref{eq:1st-ineq-bounds}.

Now let us prove that%
\begin{equation}
p_{s}^{\text{inc,unc}}(\mathcal{N}^{0},\mathcal{N}^{1})\leq p_{s}^{\text{inc}%
}(\mathcal{N}^{0},\mathcal{N}^{1}).
\end{equation}
Let $\left\{  P_{RB\rightarrow RBE^{\prime}}^{i}\right\}  _{i\in\left\{
0,1\right\}  }$ be arbitrary operators satisfying $\sum_{i}P^{i\dag}%
P^{i}=I_{RB}$. Then consider that%
\begin{equation}
\sum_{i\in\left\{  0,1\right\}  }\mathcal{N}^{i\dag}(P^{i\dag})
\mathcal{N}^{i\dag}(P^{i})  \leq\sum_{i\in\left\{  0,1\right\}
}\mathcal{N}^{i\dag}(P^{i\dag}P^{i}),
\end{equation}
as a direct consequence of the Kadison--Schwarz inequality (see \cite[Exercise~6.7]{H12book}). Then set
$\Lambda_{RB}^{i}=P^{i\dag}P^{i}$ and these operators satisfy $\Lambda
_{RB}^{i}\geq0$ for $i\in\left\{  0,1\right\}  $ and $\sum_{i\in\left\{
0,1\right\}  }\Lambda_{RB}^{i}=I_{RB}$. Thus,%
\begin{align}
&  \frac{1}{2}\left\Vert \sum_{i\in\left\{  0,1\right\}  }\mathcal{N}^{i\dag
}(P^{i\dag})  \mathcal{N}^{i\dag}(P^{i})  \right\Vert _{\infty
}\nonumber\\
&  \leq\frac{1}{2}\left\Vert \sum_{i\in\left\{  0,1\right\}  }\mathcal{N}%
^{i\dag}(P^{i\dag}P^{i})\right\Vert _{\infty}\\
&  \leq p_{s}^{\text{inc}}(\mathcal{N}^{0},\mathcal{N}^{1}),
\end{align}
where in the last line we exploit \eqref{eq:incoh-suc-prob-compare}.

Let us finally establish
\begin{equation}
p_{s}^{\text{inc}}(\mathcal{N}^{0},\mathcal{N}^{1})\leq\sqrt{p_{s}%
^{\text{coh}}(\mathcal{N}^{0},\mathcal{N}^{1})}.
\label{eq:last-ineq-prob-s-bnds}%
\end{equation}
By picking $P_{RB\rightarrow RBE^{\prime}}^{i}=\sqrt{\Lambda_{RB}^{i}}%
\otimes|0\rangle_{E^{\prime}}$, where $\Lambda_{RB}^{i}\geq0$ for
$i\in\left\{  0,1\right\}  $ and $\sum_{i\in\left\{  0,1\right\}  }%
\Lambda_{RB}^{i}=I_{RB}$, and exploiting
\eqref{eq:coh-ch-disc-suc-prob-derive}, we find that%
\begin{align}
&  2\sqrt{p_{s}^{\text{coh}}(\mathcal{N}^{0},\mathcal{N}^{1})}\nonumber\\
&  =\sup_{\substack{\left\{  P^{i}\right\}  _{i\in\left\{  0,1\right\}
}:\\\sum_{i}P^{i\dag}P^{i}=I_{RB}}}\left\Vert \sum_{i\in\left\{  0,1\right\}
}(\mathcal{N}_{A\rightarrow B}^{i})^{\dag}(P_{RB\rightarrow RBE^{\prime}}%
^{i})\right\Vert _{\infty}\\
&  \geq\left\Vert \sum_{i\in\left\{  0,1\right\}  }(\mathcal{N}_{A\rightarrow
B}^{i})^{\dag}\left(\sqrt{\Lambda_{RB}^{i}}\otimes|0\rangle_{E^{\prime}%
}\right)\right\Vert _{\infty}\\
&  =\left\Vert \sum_{i\in\left\{  0,1\right\}  }(\mathcal{N}_{A\rightarrow
B}^{i})^{\dag}\left(\sqrt{\Lambda_{RB}^{i}}\right)\right\Vert _{\infty}\\
&  \geq\left\Vert \sum_{i\in\left\{  0,1\right\}  }(\mathcal{N}_{A\rightarrow
B}^{i})^{\dag}(\Lambda_{RB}^{i})\right\Vert _{\infty},
\end{align}
where the last inequality follows because $\sqrt{\Lambda} \geq \Lambda$ for $0 \leq \Lambda \leq I$.
Since $\{\Lambda_{RB}^{i}\}_{i\in\left\{  0,1\right\}  }$ is arbitrary, we
conclude \eqref{eq:last-ineq-prob-s-bnds}\ after making use of \eqref{eq:incoh-suc-prob-compare}.

\section{Comparison of coherent and incoherent channel discrimination for generalized amplitude damping channels}

\begin{figure}[ptb]
\begin{center}
\includegraphics[
width=3.3399in
]{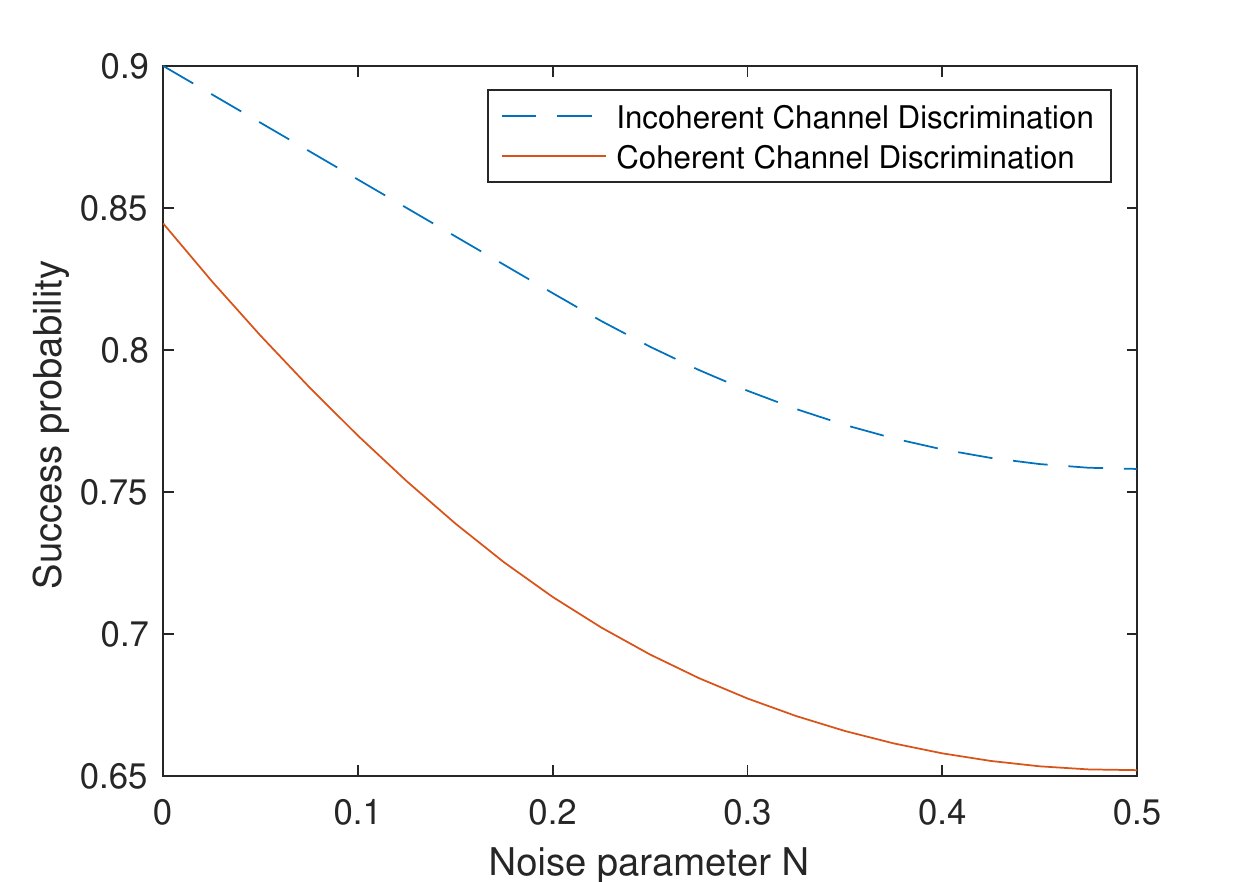}
\end{center}
\caption{Comparison of the success probabilities of coherent and incoherent channel discrimination for a generalized amplitude damping channel with  damping parameter $\gamma=0.1$ and another with damping parameter $\gamma=0.9$. The channels have the same value of the noise parameter $N$, which is varied in the plot.}%
\label{fig:ch-compare-1}%
\end{figure}

\begin{figure}[ptb]
\begin{center}
\includegraphics[
width=3.3399in
]{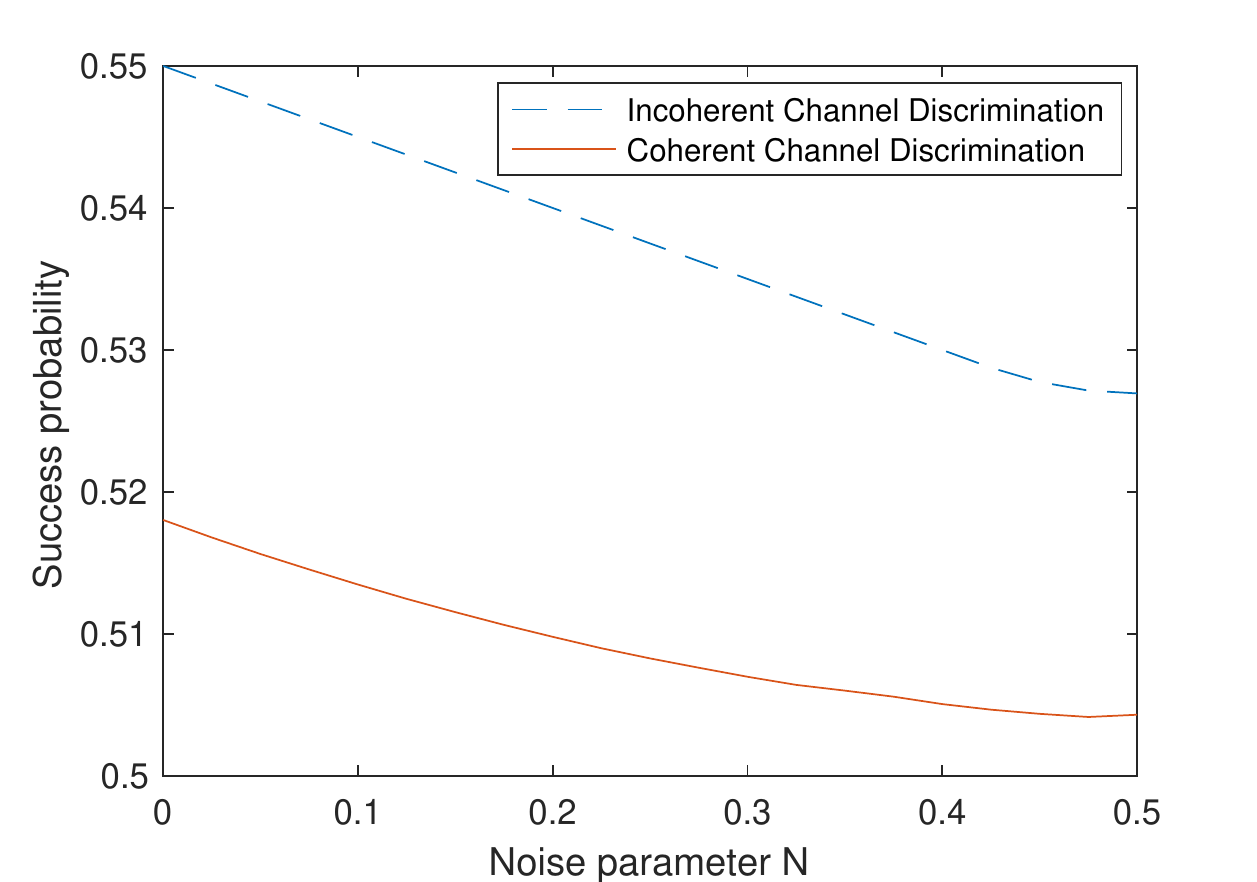}
\end{center}
\caption{Comparison of the success probabilities of coherent and incoherent channel discrimination for a generalized amplitude damping channel with  damping parameter $\gamma=0.2$ and another with damping parameter $\gamma=0.3$. The channels have the same value of the noise parameter $N$, which is varied in the plot.}%
\label{fig:ch-compare-2}%
\end{figure}

In this final appendix, I perform a comparison of the success probability of coherent and incoherent channel discrimination for generalized amplitude damping channels. The generalized amplitude damping channel is a simple model of relaxation and thermal noise that can affect a qubit \cite{book2000mikeandike}. It is governed by a damping parameter $\gamma \in [0,1]$ and a noise parameter $N \in[0,1]$. When the noise parameter $N=0$, it reduces to the standard amplitude damping channel. It is defined by the following four Kraus operators \cite{book2000mikeandike}:
\begin{align}
&\sqrt{1-N}\left(|0\rangle \langle 0| + \sqrt{1-\gamma} |1\rangle \langle 1|\right), \\
&\sqrt{\gamma(1-N)}|0\rangle \langle 1|, \\
&\sqrt{N}\left(\sqrt{1-\gamma}|0\rangle \langle 0| +  |1\rangle \langle 1|\right), \\
&\sqrt{\gamma N}|1\rangle \langle 0| .
\end{align}

Using these Kraus operators and the semi-definite programming formulation of the success probability of coherent channel discrimination from \eqref{eq:SDP-succ-prob-coh-disc}, we can calculate it for generalized amplitude damping channels (Matlab files for doing so are available with the arXiv posting of this paper). We can also calculate the success probability of incoherent channel discrimination of the same channels by using the semi-definite programming formulation of the diamond distance from \cite{Wat09} and combining with \eqref{eq:incoh-succ-prob-ch-disc}.

Figures~\ref{fig:ch-compare-1} and \ref{fig:ch-compare-2} compare the success probabilities of coherent and incoherent channel discrimination for generalized amplitude damping channels with different values of the damping parameter $\gamma$ and the noise parameter $N$.

\end{document}